\PassOptionsToPackage{unicode}{hyperref}

\PassOptionsToPackage{hyphens}{url}
\PassOptionsToPackage{dvipsnames,svgnames,x11names}{xcolor}
\documentclass[12pt]{article}
\usepackage{hyperref}
\usepackage{xr-hyper}       
\usepackage{amsmath,amssymb}
\usepackage{iftex}
\ifPDFTeX

  \usepackage[T1]{fontenc}
  \usepackage[utf8]{inputenc}
  \usepackage{textcomp} 
\else 
  \usepackage{unicode-math}
  \defaultfontfeatures{Scale=MatchLowercase}
  \defaultfontfeatures[\rmfamily]{Ligatures=TeX,Scale=1}
\fi
\usepackage{lmodern}
\ifPDFTeX\else  
\fi
\IfFileExists{upquote.sty}{\usepackage{upquote}}{}
\IfFileExists{microtype.sty}{
  \usepackage[]{microtype}
  \UseMicrotypeSet[protrusion]{basicmath} 
}{}
\usepackage{xcolor}
\usepackage{longtable,booktabs,array}
\usepackage{calc} 
\usepackage{etoolbox}
\makeatletter
\patchcmd\longtable{\par}{\if@noskipsec\mbox{}\fi\par}{}{}
\makeatother
\IfFileExists{footnotehyper.sty}{\usepackage{footnotehyper}}{\usepackage{footnote}}
\makesavenoteenv{longtable}
\usepackage{graphicx}
\makeatletter
\def\maxwidth{\ifdim\Gin@nat@width>\linewidth\linewidth\else\Gin@nat@width\fi}
\def\maxheight{\ifdim\Gin@nat@height>\textheight\textheight\else\Gin@nat@height\fi}
\makeatother
\setkeys{Gin}{width=\maxwidth,height=\maxheight,keepaspectratio}
\makeatletter
\def\fps@figure{htbp}
\makeatother
\usepackage[left=1in,top=0.8in,right=1in,bottom=0.7in,head=.2in,nofoot]{geometry}

\usepackage[font=small,labelfont=bf]{caption}
\makeatletter
\@ifpackageloaded{caption}{}{\usepackage{caption}}
\AtBeginDocument{%
\ifdefined\contentsname
  \renewcommand*\contentsname{Table of contents}
\else
  \newcommand\contentsname{Table of contents}
\fi
\ifdefined\listfigurename
  \renewcommand*\listfigurename{List of Figures}
\else
  \newcommand\listfigurename{List of Figures}
\fi
\ifdefined\listtablename
  \renewcommand*\listtablename{List of Tables}
\else
  \newcommand\listtablename{List of Tables}
\fi
\ifdefined\figurename
  \renewcommand*\figurename{Figure}
\else
  \newcommand\figurename{Figure}
\fi
\ifdefined\tablename
  \renewcommand*\tablename{Table}
\else
  \newcommand\tablename{Table}
\fi
}
\@ifpackageloaded{float}{}{\usepackage{float}}
\floatstyle{ruled}
\@ifundefined{c@chapter}{\newfloat{codelisting}{h}{lop}}{\newfloat{codelisting}{h}{lop}[chapter]}
\floatname{codelisting}{Listing}

\makeatother
\makeatletter
\@ifpackageloaded{caption}{}{\usepackage{caption}}
\@ifpackageloaded{subcaption}{}{\usepackage{subcaption}}
\makeatother

\ifLuaTeX
  \usepackage{selnolig}  
\fi
\usepackage[]{natbib}
\usepackage{bookmark}

\IfFileExists{xurl.sty}{\usepackage{xurl}}{} 
\hypersetup{
  pdftitle={Title},
  pdfauthor={Author 1; Author 2},
  pdfkeywords={3 to 6 keywords, that do not appear in the title},
  colorlinks=true,
  linkcolor={blue},
  filecolor={Maroon},
  citecolor={Blue},
  urlcolor={Blue},
  pdfcreator={LaTeX via pandoc}}

\usepackage{xcolor}
\newcommand{\anon}{1}

\usepackage{mycommands}

\newtheorem{theorem}{Theorem}
\newtheorem{lemma}{Lemma}

\newtheorem{corollary}{Corollary}

\newtheorem{example}{Example}
\newtheorem{remark}{Remark}
\newtheorem{assumption}{Assumption}

\newtheorem*{condition*}{Condition}

\newtheorem{proposition}{Proposition}

\theoremstyle{plain}
\theoremstyle{definition}

\makeatletter
\providecommand\hyper@newdestlabel[2]{}      
\providecommand\HyField@AuxAddToFields[1]{}
\providecommand\HyField@AuxAddToCoFields[2]{}
\providecommand\newlabelxx[6]{}
\makeatother

\makeatletter
\providecommand\hyper@newdestlabel[2]{}      
\providecommand\HyField@AuxAddToFields[1]{}
\providecommand\HyField@AuxAddToCoFields[2]{}
\providecommand\newlabelxx[6]{}
\makeatother

\begin{document}

\def\spacingset#1{\renewcommand{\baselinestretch}%
{#1}\small\normalsize} \spacingset{1}


\if1\anon
{
  \title{\bf Stochastic Sensitivity Analysis for Matched Observational Studies}
  \author{\\Mengqi Lin, Colin B. Fogarty and Gongjun Xu
   \hspace{.2cm} \\
   Department of Statistics, University of Michigan}
    \date{}
  \maketitle
  \date{}
} \fi

\bigskip
\begin{abstract}
Sensitivity analysis asks how strong unmeasured confounding needs to be to explain away an observational study’s conclusion. The conventional approach in matched studies conducts inference conditional upon the potential outcomes as well as both observed and unobserved confounders, and then finds the worst-case distribution for the conditional treatment assignments across all possible realizations of the unobserved confounder. The resulting worst-case allocation imagines strong, near perfect, correlations between the potential outcomes and hidden bias. We propose a stochastic sensitivity analysis that instead targets inference conditional upon potential outcomes and observed confounders while treating the hidden confounders as random with unknown conditional laws. Rather than finding the worst-case realizations for the hidden confounders, we instead determine the worst-case conditional law over a broad class of distributions. This preserves the adversarial spirit of sensitivity analysis while allowing for imperfect alignment between hidden bias and potential outcomes to a degree controlled by a scalar sensitivity parameter. We consider restrictions to both an interpretable class with no parametric assumptions and a Bernoulli class of conditional laws. Design sensitivity calculations and real-data demonstrations illustrate that allowing for even a small degree of stochasticity can materially increase reported robustness to hidden bias relative to the conventional approach.
\end{abstract}

\noindent%
\vfill




\newpage
\spacingset{1.8} %
\section{Introduction}

\subsection{A motivating example}\label{subsec:mt_example} 

In \citeauthor{ham64}'s \citeyearpar{ham64} study of smoking and lung cancer, 36,975 male nonsmokers were matched to men who smoked 20 or more cigarettes per day on a broad set of demographic, occupational, and health characteristics. In these data, there were 122 discordant matched pairs in which exactly one subject died of lung cancer:  in 110 pairs the smoker died of lung cancer, whereas in 12 pairs the nonsmoker died of lung cancer  \citep{ros87,Rosenbaum2002}. While this striking difference provides evidence for an association, it does not provide direct evidence for a causal effect of smoking without further assumptions: as treatment was not randomized, the conclusion remains vulnerable to unmeasured confounding. One proposed source of hidden bias in the early studies of smoking and lung cancer is genetic variation affecting both smoking behavior and lung cancer risk \citep{fisher_lung_1958}. Recent genetic studies have reported that variants in region q24--25.1 of chromosome 15 are associated with smoking and nicotine dependence \citep{Thorgeirsson2008,Weiss2008,Saccone2009, Lassi2016}. Genome-wide association studies have also linked variants in this same region to lung-cancer risk \citep{Hung2008,Amos2008,Thorgeirsson2008}. A risk allele in this region represents a potential unmeasured confounder associated with both smoking behavior and lung-cancer risk.


Rather than assuming away hidden bias, sensitivity analysis asks how strong unmeasured confounding would have to be to overturn the study's conclusion. In studies such as \citet{ham64} that use matching to adjust for overt biases, methods derived under the model of \citet{ros87} are the standard choice for probing robustness to unobserved confounding. This conventional sensitivity model considers an unmeasured confounder \(U_{ij}\in[0,1]\) and, at sensitivity level \(\Gamma\ge 1\), assumes that subjects $(i,j)$ and $(i,\ell)$ in matched set $i$ differ in their odds of exposure to treatment by $\Gamma^{U_{ij}-U_{i\ell}}$ because of this hidden confounder. Larger values of \(\Gamma\) therefore permit hidden bias to have a stronger impact on the assignment mechanism. In the smoking example, one may interpret \(U_{ij}=1\) as indicating that a subject carries the risk allele at that particular variant, and \(U_{ij}=0\) otherwise.
Then a carrier may be up to \(\Gamma\) times more likely to be the smoker than a noncarrier in the same matched set. 

Under this model, the conventional sensitivity analysis proceeds to find an upper bound on the $p$-value by identifying the worst-case values the unmeasured confounders $U_{ij}$ for all individuals in the study. In Hammond's study, the conventional approach proceeds as if, \textit{for all 122 discordant pairs}, the subject who died of lung cancer had $U_{ij}=1$ and the matched partner had $U_{i\ell}=0$. For the genetic variant, this would mean that all of the individuals who died from lung cancer in discordant pairs carried the risk allele, and that none of the individuals who did not die from lung cancer carried the risk allele. Such an alignment is stronger than the genetic evidence suggests: genetic associations reported by \citet{Amos2008}
and \citet{Thorgeirsson2008} between lung cancer risk and risk allele presence are modest on the odds-ratio scale (for instance, a difference of 30\% between those with a single risk allele and those without any), but not deterministic. Subjects who die of lung cancer would be more likely, but not certain, to carry that allele. In that case, some matched pairs would likely contain two noncarriers, some two carriers, and some would contain a carrier who did not die of lung cancer. The conventional sensitivity analysis instead places the hidden confounder in the most adverse configuration, effectively aligning it with the outcome perfectly.

\subsection{Overview and contributions}\label{subsec:overview}
A growing literature has sought to relax this deterministic worst-case analysis. One line of work does so by allowing the magnitude of hidden bias to vary across matched sets for matched observational studies. 
For matched pairs, \citet{hasegawa_sensitivity_2017} showed that this worst-case  analysis can be calibrated to average rather than worst-case hidden bias. \citet{fog19extend} introduced an extended sensitivity analysis that simultaneously bounds maximal and typical bias. Recently, \cite{wu25} studied quantiles of hidden biases. These papers primarily address heterogeneity in the magnitude of hidden bias across matched sets, but they do not directly relax the conventional deterministic worst-case allocation of hidden bias. 
Another line of research treats the unmeasured confounder as a random latent variable and specifies treatment and outcome models involving that variable \citep{Rosenbaum1983,imb03,carnegie2016assessing,Dorie2016,zhang2020calibrated}. These approaches are typically more model-based and often rely on parametric or otherwise low-dimensional assumptions on the latent confounder and/or on the treatment and outcome models.



This paper develops a stochastic sensitivity analysis for matched observational studies.  Rather than considering  inference valid for any realization of the potential outcomes, observed confounders and unobserved confounder, we instead consider valid inference given the potential outcomes and observed covariates while treating the unobserved covariate as random, considering the worst-case conditional \textit{distribution} for the unobserved covariates given the potential outcomes over a broad class of distributions. Rather than committing to a single fully specified latent-variable treatment and outcome model, we retain the finite-population, randomization-based framework that preserves the worst-case spirit of sensitivity analysis for matched observational studies. The conventional sensitivity analysis is recovered when the class is allowed to include degenerate distributions concentrated on the most adverse configurations. Our approach allows for departures from this most adverse configuration through an additional sensitivity parameter, roughly viewed as the degree of stochasticity in the conditional law.

Once inference only conditions upon the potential outcomes and measured covariates, finding the required worst-case \(p\)-value becomes an optimization over distributions for the unobserved confounders rather than over fixed vectors of them. Nonetheless, we show that the separable algorithm of \citet{gas00}, originally introduced under the conventional approach, remains the appropriate principle in this stochastic setting. Specifically, within each matched set, we first identify the distributions that maximize the mean of the set-specific statistic and then, among all maximizers, break ties by choosing one with the largest  variance. 
We show that this principle yields the exact least-favorable upper tail for two-point statistics and an asymptotically conservative upper tail for general statistics.

We consider a broad class of possible distributions for the hidden confounders.
If this class were left completely unrestricted, it would include degenerate point masses on the most adverse configurations and would therefore collapse to the conventional deterministic worst-case analysis. To move beyond that extreme, we impose a mean-band restriction which requires the expected values of the hidden confounders to lie away from the extreme endpoints.
This restriction  rules out the most extreme distributions without committing to a specific parametric family. 
The strength of this distributional restriction is controlled by the additional sensitivity parameter \(g\in[0,1/2]\). The usual parameter \(\Gamma\) controls the magnitude of hidden bias, while \(g\) controls how close the conditional distribution of the hidden confounder may come to deterministic worst-case alignment. When \(g=0\), the class includes the degenerate distributions that recover the conventional sensitivity analysis; as \(g\) increases, these near-deterministic adverse distributions are progressively ruled out.
Within this class, we further study two interpretable subclasses. The first is a two-group subclass, in which the joint distribution for the hidden confounders within a matched set follows a mixture of product distributions, with each product comprised of at most two distinct marginals. Motivated by the genetic illustration in the Hammond example, we also consider a Bernoulli subclass, in which the hidden confounders take the endpoint values with different probabilities. In both cases, the least favorable distributions retain the conventional  top-$k$ structure (see Section~\ref{sec:rosen} and Section~\ref{sec:classes} for details).

This framework allows us to assess how robustness statements arising from the conventional approach change as stochasticity in the conditional distribution of the unmeasured confounder is introduced. In Section~\ref{sec:power}, design sensitivity
calculations show that small positive values of \(g\) can substantially
increase the value of \(\Gamma\) at which rejection is sustained, relative to
\(g=0\) of the conventional sensitivity analysis. The applications in
Section~\ref{sec:application} illustrate the same phenomenon in finite samples. For Hammond's smoking example, under the conventional approach allowing for perfect alignment between the unobserved confounder and the potential outcomes, two individuals in the same matched set would need to differ in their odds of smoking by a factor of 5.6 in order to overturn the finding of smoking causing lung cancer. If one restricts the conditional distribution of $U_{ij}$ given the potential outcomes and covariates to $\operatorname{Bernoulli}(p_{ij})$ with $p_{ij}\in [g,1-g] = [0.1,0.9]$ for each $(i,j)$ and then finds the worst-case success probabilities $p_{ij}$ for all $(i,j)$, two individuals in the same matched set would need to differ in their odds of smoking by a factor of 14.5 to overturn the finding of smoking causing lung cancer.


The rest of the paper is organized as follows. Section~\ref{sec:background} introduces the setup for matched observational studies and reviews the conventional sensitivity model. Section~\ref{sec:justify-mean} establishes exact and asymptotic justification for the separable algorithm under stochastic confounding. Section~\ref{sec:classes} develops concrete classes of distributions and characterizes their least favorable distributions. Section~\ref{sec:power} compares stochastic and deterministic sensitivity analysis through design sensitivity. Section~\ref{sec:application} presents the reanalyses of Hammond's data along with another study investigating the impact of binge drinking on blood pressure.

\section{Background: matched design and  sensitivity analysis}\label{sec:background}
\subsection{Matched design and sensitivity model}
We consider an observational study partitioned into $I$ matched sets on the basis of observed pre-treatment covariates. Matched set $i=1,...,I$ contains $n_i\ge 2$ units indexed by $j=1,...,n_i$, with total sample size $N=\sum_{i=1}^I n_i$. Each matched set contains one treated unit and $n_i-1$ controls; extensions to full matching are straightforward (e.g., \citet[][Ch.~4, Prob.~4.12]{Rosenbaum2002}). For unit $(i,j)$, let $Z_{ij}\in\{0,1\}$ indicate treatment assignment, so that $\sum_{j=1}^{n_i} Z_{ij}=1$; let $\bx_{ij}\in\mathbb{R}^p$ denote the observed covariates used in matching; and let $(r_{Tij},r_{Cij})$ be the potential outcomes under treatment and control. Under SUTVA \citep{rub74}, the observed outcome is $R_{ij}=Z_{ij}r_{Tij}+(1-Z_{ij})r_{Cij}.$
Let $\bZ=(Z_{11},\ldots,Z_{In_I})^\top$ and $\bZ_i=(Z_{i1},\ldots,Z_{in_i})^\top$, and use the same stacking convention for $\bR$, $\br_T$, $\br_C$, and other set-level quantities. Define
\[
\Omega
:=
\Bigl\{\bz\in\{0,1\}^N:\ \sum_{j=1}^{n_i} z_{ij}=1 \text{ for all } i=1,\dots,I\Bigr\},
\qquad
\cZ:=\{\bZ\in\Omega\},
\]
where $\Omega$ is the set of treatment assignments consistent with the matched design and $\cZ = \{\bZ \in \Omega\}$ is the event that the realized assignment belongs to this set. We work in the finite population framework, conditioning on $\mca:=\{\br_T,\br_C,\bx\}$.

Let $G_{ij}$ denote an unmeasured confounder for unit $(i,j)$, and collect these as
$\bG_i=(G_{i1},\ldots,G_{in_i})^\top$ and $\bG=(\bG_1,\dots,\bG_I)$. Write $\pi_{ij}=\P(Z_{ij}=1\mid \mca,\bG)$ for the treatment assignment probability prior to matching, and assume that $Z_{ij}$ are independent given $(\mca, \bG)$. We consider the
sensitivity model
\begin{equation}\label{eq:sensmodel}
\log\!\left(\frac{\pi_{ij}}{1-\pi_{ij}}\right)
=
\kappa(\bx_{ij})+\log(G_{ij}),
\end{equation}
where $\kappa(\cdot)$ is an unknown function of the observed covariates. The sensitivity parameter
$\Gamma\ge1$ controls the strength of unmeasured confounding by restricting the support of $G_{ij}$ to
$[1,\Gamma]$. The conventional sensitivity model is recovered as the special case $G_{ij}=\Gamma^{U_{ij}}$ for some $U_{ij}\in[0,1]$ \citep{ros87}.

Under the idealized setting where matching is exact, \(\bx_{ij}=\bx_{ij'}\) for all \(j,j'\) within set \(i\). Conditioning on \(\cZ\) then removes the nuisance term \(\kappa(\bx_{ij})\) within set \(i\), yielding the conditional treatment probabilities
\begin{equation}\label{eq:rho}
\varrho_{ij}(\bG_i)
:= \P(Z_{ij}=1\mid \mca,\cZ,\bG)
=
\frac{G_{ij}}{\sum_{\ell=1}^{n_i} G_{i\ell}}.
\end{equation}
Thus, after conditioning on the matched design, the treatment assignment probabilities within set $i$ depend only on $\bG_i$. 
The sensitivity parameter $\Gamma$ quantifies the magnitude of hidden bias: within a matched set, two units with the same observed covariates may differ in their conditional treatment probabilities by at most a factor of $\Gamma$ because of unmeasured confounding.
When $\Gamma=1$, we have $G_{ij}\equiv 1$ and hence $\varrho_{ij}=1/n_i$ for all $j$. As $\Gamma$ increases, the restriction $G_{ij}\in[1,\Gamma]$ permits greater heterogeneity within $\bG_i$, allowing $\varrho_{ij}$ to depart further from $1/n_i$ and thereby representing stronger potential hidden bias.

In this work, we test Fisher's sharp null hypothesis of no causal effect,
\begin{equation}\label{eq:sharpnull}
H_F:\ r_{Tij}=r_{Cij}\quad \forall\, i,j.
\end{equation}
Under $H_F$, all missing potential outcomes are imputed by the observed outcomes. Extensions to other null hypotheses are available; see \citet[][Section~5]{Rosenbaum2002} and \citet[][Sections~2.4--2.5]{designofobs}. Let $\bR_{\bZ}$ denote the vector of observed outcomes under assignment $\bZ$. We consider test statistics of the form
\[
T(\bZ)=\bZ^\top \bq=\sum_{i=1}^I\sum_{j=1}^{n_i} Z_{ij}\,q_{ij},
\qquad \bq=q(\bR_{\bZ}),
\]
which include many common test statistics through an appropriate choice of the score function $q(\cdot)$ \citep{Rosenbaum2002}. Under \eqref{eq:sharpnull}, the score vector is fixed across $\Omega$: for any $\zz\in\Omega$, $q(\bR_{\zz})=q(\bR_{\bZ})$. Hence $T(\bz)=\bz^\top\bq$ is well defined for every $\bz\in\Omega$. Assuming the sensitivity model \eqref{eq:sensmodel} holds at a given $\Gamma$, the conditional right-tail probability of $T$ is
\begin{equation}\label{eq:tailprob}
\P\!\bigl(T(\bZ)\ge a \,\big|\, \mca,\cZ,\bG\bigr)
=
\sum_{\bz\in\Omega} \1\!\bigl\{\bz^\top\bq\ge a\bigr\}\,
\P(\bZ=\bz\mid \mca,\cZ,\bG),
\end{equation}
where $\P(\bZ=\bz\mid \mca,\cZ,\bG)$ is induced by the within-set probabilities in \eqref{eq:rho}. In particular, when $a$ is taken to be the observed value of the test statistic $T_{\mathrm{obs}}$, \eqref{eq:tailprob} is the conditional right-tail $p$-value.

Although this \(p\)-value is computable  under \(H_F\) once the treatment assignment distribution is specified, it is unavailable to the analyst when \(\Gamma>1\) because it depends on the unobserved confounders \(\bG\). Sensitivity analysis posits a class of mechanisms for hidden bias indexed by \(\Gamma\) and then maximizes \eqref{eq:tailprob} with $a = T_{\mathrm{obs}}$ over that class, the resulting upper bound is the worst-case \(p\)-value. The study's robustness to hidden bias can then be summarized by the largest value of \(\Gamma\) for which rejection of \(H_F\) at level \(\alpha\) persists.

\subsection{Conventional sensitivity analysis and the separable algorithm}\label{sec:rosen}
Fix $\Gamma\ge 1$. The conventional sensitivity analysis treats the unmeasured confounders $G_{ij}\in[1,\Gamma]$ as fixed and upper-bounds \eqref{eq:tailprob} by maximizing it over all feasible configurations of $\bG$:
\begin{equation}\label{eq:rosen_worstp}
\sup_{\bG\in[1,\Gamma]^N}\;\P\!\bigl(T(\bZ)\ge a \,\big|\, \mca,\cZ,\bG\bigr).
\end{equation}
Thus, the conventional analysis targets validity conditional on \((\mca,\cZ,\bG)\), uniformly over all fixed realizations of the hidden confounders. 
With $a=T_{\obs}$, this results in a conservative $p$-value that is valid uniformly over all fixed $\bG\in[1,\Gamma]^N$. 

As noted by \citet{gas00}, the optimization in \eqref{eq:rosen_worstp} is generally nonseparable, except for two-point statistics (see Section~\ref{sec:two-point}). That is, the least favorable configuration of \(\bG\) cannot, in general, be obtained by solving separate optimization problems set by set and then combining the results. 
For large samples, \citet{gas00} proposed a separable approximation based on a normal approximation to the distribution of \(T\).  
Writing $T=\sum_{i=1}^I T_i,$ $T_i=\sum_{j=1}^{n_i} Z_{ij}q_{ij},$
their algorithm proceeds as follows. In matched set \(i\), it first chooses \(\bG_i\) to maximize the conditional mean of \(T_i\), and then, among all maximizers, selects the one that maximizes the conditional variance of \(T_i\). The resulting expectations and variances are then aggregated across matched sets to form the approximated upper bound.

To describe this procedure, let
\[
\mu_i(\bG_i)
:= \E(T_i\mid \mca,\cZ,\bG_i)
= \sum_{j=1}^{n_i} \varrho_{ij}(\bG_i)\,q_{ij}
\]
denote the conditional mean of \(T_i\) given \(\bG_i\). The corresponding maximized mean is
\begin{equation}\label{eq:mu_given_G}
(\mu_i)^R
:=\max_{\bG_i\in[1,\Gamma]^{n_i}}\ \mu_i(\bG_i)
=
\max_{\bG_i\in[1,\Gamma]^{n_i}}
\sum_{j=1}^{n_i}
\frac{G_{ij}}{\sum_{\ell=1}^{n_i}G_{i\ell}}\ q_{ij}.
\end{equation}
After relabeling the units so that \(q_{i1}\ge \cdots \ge q_{in_i}\), a maximizer has a top-\(k\) form: there exists \(k\in\{1,\dots,n_i-1\}\) such that
\begin{equation}\label{eq:rosen_topk}
    G_{ij}^\ast=
\begin{cases}
\Gamma, & j\le k,\\
1, & j>k.
\end{cases}
\end{equation}
The maximizing value of $k$ need not be unique. When several values of $k$ attain the same maximum mean $(\mu_i)^R$, the method breaks ties by choosing the configuration that maximizes the conditional variance
\[
\nu_i^2(\bG_i)
:=\Var(T_i\mid \mca,\cZ,\bG_i)
=\sum_{j=1}^{n_i}\varrho_{ij}(\bG_i)\,q_{ij}^2
-\left\{\sum_{j=1}^{n_i}\varrho_{ij}(\bG_i)\,q_{ij}\right\}^2.
\]
Let $(\nu_i^2)^R$ denote this maximized variance. Under mild regularity conditions, the right-tail probability in \eqref{eq:tailprob} is asymptotically upper-bounded by the tail probability of a normal random variable with mean $\sum_{i=1}^I(\mu_i)^R$ and variance $\sum_{i=1}^I(\nu_i^2)^R$.
This approximation was shown to recover the correct optimum up to negligible error \citep{gas00}.

Because the objective in \eqref{eq:rosen_worstp} optimizes over all fixed confounder realizations after conditioning on the potential outcomes, the conventional analysis can be driven by configurations in which hidden bias is nearly perfectly aligned with the potential outcomes.
As an illustration, consider matched pairs with \(n_i\equiv 2\) and \(q_{i1}\ge q_{i2}\). The maximum in \eqref{eq:mu_given_G} is attained by the extreme configuration \(G_{i1}=\Gamma\) and \(G_{i2}=1\). Thus, the conventional sensitivity analysis assigns the largest confounder value to the unit with the larger score and the smallest  confounder value to the unit with the smaller score. 
In the motivating example of Section~\ref{subsec:mt_example}, this means that, in each discordant pair, the subject who died of lung cancer is assigned the larger confounder value, as if that subject carried the risk allele while the matched subject did not. 

\subsection{From deterministic worst-case bias to stochastic confounding}
\label{subsec:stochastic-G}
By maximizing over all fixed realizations of \(\bG\) in (\ref{eq:rosen_worstp}), the conventional sensitivity analysis provides valid inference for any conditioning set \((\mca,\cZ,\bG)\). Were one to additionally imagine that $\mathbf{G}$ has a conditional distribution given $(\mca, \cZ)$, the conventional approach would also confer validity conditional upon $(\mca, \cZ)$. In what follows we instead target inference conditional on \((\mca,\cZ)\) directly while treating the hidden confounders as random with an unknown conditional distribution. Thus, the adversary no longer chooses a single worst-case realization of \(\bG\), but rather a conditional distribution for \(\bG\) within a specified class.

For each matched set \(i\), let \(\cP_i\) denote a class of candidate laws for \(\bG_i=(G_{i1},\ldots,G_{in_i})^\top\), each supported on $[1,\Gamma]^{n_i}$. 
As in the usual randomization framework for matched observational studies, we treat distinct matched sets as independent blocks. The class of candidate laws for $\bG$ can thus be denoted as $\cP
:=
\{P=\bigotimes_{i=1}^I P_i:\ P_i\in\cP_i\}$.
Note that, if each \(\cP_i\) were allowed to contain all distributions on \([1,\Gamma]^{n_i}\), then point-mass distributions would be included and the resulting  bound would collapse to the conventional sensitivity analysis. The stochastic formulation becomes distinct from the conventional analysis once \(\cP_i\) rules out some such degenerate conditional distributions.

Because \(\bG\) is no longer treated as fixed in the analysis, the optimization now ranges over distributions \(P\in\cP\) for \(\bG\), rather than over fixed vectors \(\bG\in[1,\Gamma]^N\).
 For any threshold \(a\), the corresponding
worst-case right-tail probability is
\[
\sup_{P\in\cP}\;\P_P\!\bigl(T(\bZ)\ge a \,\big|\, \mca,\cZ\bigr)
=
\sup_{P\in\cP}\;
\E_{P}\!\Big[\P\!\bigl\{T(\bZ)\ge a \,\big|\, \mca,\cZ,\bG\bigr\}\Big].
\]
Here \(\P_P\) denotes the joint law obtained by first drawing
\(\bG\sim P\) conditional on \((\mca,\cZ)\), and then drawing the treatment
assignment according to the conditional probabilities in \eqref{eq:rho}. Thus,
the guarantee is uniform over conditional distributions \(P\in\cP\), rather
than over every fixed realization of \(\bG\). When \(a=T_{\obs}\), the
quantity is the stochastic worst-case \(p\)-value.

As in the deterministic optimization problem \eqref{eq:rosen_worstp}, direct optimization over \(\cP\) is generally nonseparable across matched sets except for two-point statistics (see Section~\ref{sec:two-point}).
Motivated by \citet{gas00}, we consider the same separable large-sample principle under our stochastic formulation, based on the asymptotic Gaussian behavior of \(T=\sum_{i=1}^I T_i\). Specifically, in each matched set, we first choose a distribution that maximizes the mean of \(T_i\), and among all such maximizers choose one that maximizes the variance of \(T_i\). Section~\ref{sec:justify-mean} shows that, although originally developed for deterministic confounder configurations, this separable algorithm remains valid under our stochastic formulation.

We now describe this separable algorithm for our stochastic formulation. 
For each \(P_i\in\cP_i\), denote the induced treatment probabilities as $\varrho_{ij}(P_i) := \E_{P_i}\!\bigl[\varrho_{ij}(\bG_i)\bigr]$.
After integrating over \(\bG_i\), we have $\mu_i(P_i)
:=
\E_{P_i}(T_i\mid \mca,\cZ)
=
\sum_{j=1}^{n_i}\varrho_{ij}(P_i)\,q_{ij},$
and $\nu_i^2(P_i)
:=
\Var_{P_i}(T_i\mid \mca,\cZ)
=
\sum_{j=1}^{n_i}\varrho_{ij}(P_i)\,q_{ij}^2
-
\mu_i(P_i)^2.$
In matched set \(i\), we solve
\begin{equation}\label{eq:obj_mean}
\sup_{P_i\in\cP_i}\; \mu_i(P_i),
\end{equation}
and, among all maximizers of \eqref{eq:obj_mean}, select a distribution maximizing \(\nu_i^2(P_i)\). 
A degenerate case arises when \(q_{i1}=\cdots=q_{in_i}\). Since
\(\sum_{j=1}^{n_i} Z_{ij}=1\), we then have $T_i=q_{i1}$
almost surely under every distribution \(P_i\). Hence $\mu_i(P_i)=q_{i1}$ and $\nu_i^2(P_i)=0$
for all \(P_i\in\cP_i\). The optimization is therefore trivial in
this case. Therefore, without loss of generality, we may assume that \(q_{i1}>q_{in_i}\) throughout this paper.
Section~\ref{sec:classes} studies several interpretable choices of the classes \(\cP_i\) and derives the corresponding least favorable distributions.

\section{Justification of the separable algorithm
}\label{sec:justify-mean}
In this section, we justify the separable algorithm as a valid least-favorable construction. The justification has two parts. For two-point statistics, the mean maximization step is exactly least favorable for the upper tail. For general statistics, the algorithm yields an asymptotically conservative upper tail under regularity conditions.
\subsection{Two-point statistics and exact upper tail bound}\label{sec:two-point}
For a broad class of test statistics for which each $T_i$ takes at most two values, the separable algorithm  yields the exact least-favorable distribution for the upper tail. Specifically, suppose that
\begin{equation}\label{eq:two-value-T}
T_i \;=\; a_{i2} \;+\; (a_{i1}-a_{i2})\sum_{j=1}^{n_i} Z_{ij}\,\1\{q_{ij}=a_{i1}\},\quad a_{i1} > a_{i2}.
\end{equation}
Statistics of such form include the class of sign-score statistics
\citep[\S4.3--4.4]{ros88,Rosenbaum2002}. This representation covers, for example, all test statistics in matched-pairs designs, as well as several widely used statistics under matching with multiple controls, including the Mantel--Haenszel statistic for binary outcomes.
The next proposition shows that for these statistics, choosing in each matched set a distribution that maximizes \(\mu_i(P_i)\) is exactly least favorable for the entire upper tail.
\begin{proposition}\label{prop:sum-fsd}
Consider a class of candidate distributions \(\cP\) and a test statistic \(T=\sum_{i=1}^I T_i\), where each \(T_i\) is of the form \eqref{eq:two-value-T}. For each matched set \(i\), let
\[
P_i^\ast \in \argmax_{P_i\in \cP_i}\; \mu_i(P_i),
\]
and define \(P^\ast:=\bigotimes_{i=1}^I P_i^\ast\). Then, for every \(a\in\mathbb{R}\),
\begin{equation}\label{eq:tail-prob-bound}
\sup_{P\in\cP}\;
\P_P\!\bigl(T(\bZ)\ge a \mid \mca,\cZ\bigr)
=
\P_{P^\ast}\!\bigl(T(\bZ)\ge a \mid \mca,\cZ\bigr).
\end{equation}
\end{proposition}
\noindent
Proposition~\ref{prop:sum-fsd} shows that, for two-point statistics, maximizing \(\mu_i(P_i)\) set by set yields the exact worst-case right-tail probability, and hence the exact worst-case \(p\)-value. This result is also the stochastic analogue of the corresponding exact separability result in the conventional sensitivity analysis for two-point statistics \citep{ros88}. The difference is that, here, the optimization ranges over distributions \(P\) for the hidden confounders rather than over fixed confounder configurations.

\subsection{General statistics and asymptotic upper tail bound}

Proposition~\ref{prop:sum-fsd} shows that, for two-point statistics, the separable algorithm yields exactly the least favorable upper tail.  For general statistics,
such finite-sample exactness need not hold. In the following, we show that the separable distribution returned by the separable algorithm produces asymptotically conservative upper tail probabilities.

We introduce the following notations and assumptions for our theoretical result.
Consider a class of distributions \({\cP} = \bigotimes_{i=1}^I \cP_{i}\).
For each \(i\), let \(P_{i}^\ast \in \cP_i\) be the distribution returned by the separable algorithm, and let $P^\ast:=\bigotimes_{i=1}^I P_{i}^\ast.$ Write $\mu_{i}^\ast:=\mu_i(P_{i}^\ast)$ and  $(\nu_{i}^\ast)^2:=\nu_{i}^2(P_i^*)$.
For any distribution $P=\bigotimes_{i=1}^I P_{i}\in {\cP}$,  write $\mu_{i}:=\mu_i(P_{i}),$ $\nu_{i}^2:=\nu_{i}^2(P_i).$
By construction,
$
\mu_{i}^\ast\ge \mu_{i},
$ and if $\mu_{i}^\ast=\mu_{i},$
 then $(\nu_{i}^\ast)^2\ge \nu_{i}^2.$

\begin{assumption}\label{assump:asym-separable}
For any \(P\in{\cP}\)  and  \(P^\ast\) as defined above, the following holds.
\begin{enumerate}
\item[(i)] There exist constants \(\zeta>0\), \(M<\infty\), and
\(\underline\nu^2>0\) such that, for any sufficiently large \(I\),
$
\frac{1}{I}\sum_{i=1}^I \nu_i^2\ge \underline\nu^2
$, $
\frac{1}{I}\sum_{i=1}^I (\nu_i^\ast)^2\ge \underline\nu^2,
$
$
\frac{1}{I}\sum_{i=1}^I
\E_P\left[
\left|T_i-\mu_i\right|^{2+\zeta}
\right]\le M,$
and $
\frac{1}{I}\sum_{i=1}^I
\E_{P^\ast}\left[
\left|T_i-\mu_i^\ast\right|^{2+\zeta}
\right]\le M.
$

\item[(ii)] There exists a constant \(\bar\nu^2<\infty\) such that, for all
sufficiently large \(I\) and all \(i=1,\ldots,I\),
$\nu_i^2\le \bar\nu^2$ and $(\nu_i^\ast)^2\le \bar\nu^2.$ Let $A_I(P):=\{i:(\nu_i^\ast)^2<\nu_i^2\}$ and $\pi_I(P):={|A_I(P)|}/{I}.$
There exists a constant \(\delta>0\) such that, for any sufficiently
large \(I\), 
$
\frac{1}{I}\sum_{i=1}^I(\mu_i^\ast-\mu_i)
\ge
\delta\,\pi_I(P).
$
\end{enumerate}
\end{assumption}

Assumption 1 is the stochastic analogue of the regularity
conditions used to justify the separable approximations in the
conventional  analysis \citep{gas00}. Specifically, the conditions in (i) are Lyapunov-type conditions that ensure the Gaussian
approximations are valid. Condition (ii) is needed for the separable algorithm \citep[Section~4]{gas00}, which rules out competing distributions that are close to \(P^\ast\) in mean while having larger variance than \(P^\ast\) on a nonnegligible
fraction of matched sets.

\begin{theorem}\label{thm:asym-tail}
Under Assumption~\ref{assump:asym-separable}, and for every fixed \(c>0\), let $a
:=
\sum_{i=1}^I \mu_{i}
+
c(\sum_{i=1}^I \nu_{i}^2)^{1/2}$,
then, for every \(\epsilon>0\), there exists \(I^\ast\) such that for all \(I\ge I^\ast\),
\begin{equation}\label{eq:asym-tail-bound}
\P_{P^\ast}\!\bigl(T(\bZ)\ge a\mid \mca,\cZ\bigr)
\;\ge\;
\P_{P}\!\bigl(T(\bZ)\ge a\mid \mca,\cZ\bigr)
-\epsilon.
\end{equation}
Moreover, define the Gaussian right-tail \(p\)-value computed under \(P^\ast\) by
\[
p^\ast(T_{\obs})
:=
1-\Phi\left(
\frac{
T_{\obs}-\sum_{i=1}^I \mu_i^\ast
}{
\sqrt{\sum_{i=1}^I(\nu_i^\ast)^2}
}
\right).
\]
If Fisher's sharp null \(H_F\) holds and \(P\) is the true distribution of the
hidden confounders, then, for every fixed \(\alpha\in(0,1/2)\),
\[
\limsup_{I\to\infty}
\P_{P}\!\big(
p^\ast(T_{\obs})\le \alpha
\mid \mca,\cZ\big)
\le
\alpha.
\]
\end{theorem}

\noindent
Theorem~\ref{thm:asym-tail} shows that the upper tail under \(P^\ast\) is asymptotically no smaller than that under any \(P\in\cP\), up to an arbitrarily small error. That is, for any candidate class \(\cP\), the separable algorithm yields an asymptotically conservative upper tail over that class. Moreover, under $H_F$, if \(P\) is the true conditional distribution for the hidden confounders,  the Gaussian approximated \(p\)-value computed under \(P^\ast\) controls the type-I error asymptotically. 

\section{Interpretable distribution classes and least favorable distributions}
\label{sec:classes}
 
Our initial results in Section~\ref{sec:justify-mean} on the stochastic formulation have left the classes \(\cP_i\) abstract. We now construct concrete choices of \(\cP_i\) and derive their least favorable distributions using the separable algorithm.  Recall that $\cP_i$ is the candidate class for the conditional law of $\bG_i$ given $(\mca, \cZ)$. Since the sensitivity model \eqref{eq:sensmodel} lets \(G_{ij}\) influence treatment assignment, conditioning on \(\cZ\) may induce dependence among \(G_{i1},\ldots,G_{in_i}\).
To account for such dependence, we consider general families of mixtures of product laws, allowing arbitrarily many mixture components. Such families, as is well known in the literature, provide a flexible way to capture arbitrary dependence under mild conditions \citep{mclachlan2000finite,kolda09,banerjee13}. 
Furthermore, as illustrated in the following examples, the mixture representation is naturally compatible with our sensitivity model setting. 

\begin{example}\label{Example-1}
   Let $\cP([1,\Gamma])$ denote the set of laws on $[1, \Gamma]$ and $\Omega_i := \left\{\bz_i\in\{0,1\}^{n_i}: \sum_{j=1}^{n_i}z_{ij}=1
\right\}$. Suppose that the coordinates of \(\bG_i\) are independent conditional on \(\mca\). By (\ref{eq:sensmodel}) and independence between treatment assignments given $(\mca,\bG)$, Bayes' rule implies that, for each
fixed \(\bz_i\in\Omega_i\),  $\P(\bG_i\in\cdot\mid \mca,\bZ_i=\bz_i)$ is a product law:
\[
\P(\bG_i\in\cdot\mid \mca,\cZ,\bZ_i=\bz_i)
=
\P(\bG_i\in\cdot\mid \mca,\bZ_i=\bz_i)=
\bigotimes_{j=1}^{n_i}Q_{ij}^{\bz_i},
\]
for some 
\(Q_{ij}^{\bz_i} \in \cP([1, \Gamma])\); see Section~\ref{app:mixture-generative} of the web-based
supporting material for the proof. Moreover, since
\begin{equation}\label{eq:law_decomp}
    \begin{aligned}
    \P(\bG_i\in \cdot \mid \mca,\cZ)
&=
\sum_{\bz_i\in \Omega_i} \;\P(\bG_i\in \cdot,\,\bZ_i = \bz_i \mid \mca,\cZ) \\
&= \sum_{\bz_i\in \Omega_i}\; \P(\bG_i\in \cdot\mid \mca,\cZ, \bZ_i=\bz_i)\;\P(\bZ_i=\bz_i\mid \mca,\cZ),
\end{aligned}
\end{equation}
and $\sum_{\bz_i\in \Omega_i} \P(\bZ_i=\bz_i\mid \mca,\cZ) = 1$,
\(\P(\bG_i\in\cdot\mid \mca,\cZ)\) can be viewed as a mixture of product laws, with mixing weights
\(\P(\bZ_i=\bz_i\mid \mca,\cZ)\). Conditional independence between the coordinates of \(\bG_i\) given $\mca$ arises naturally in many generative models; for instance, it would hold if one views the tuples $(Z_{ij}, G_{ij}, \bx_{ij}, r_{Tij}, r_{Cij})$ as $iid$ draws from a superpopulation.
\end{example}
\begin{example}
In Example \ref{Example-1}, dependence among \(G_{i1},\ldots,G_{in_i}\) arises after averaging over the possible assignment vectors \(\bz_i\in\Omega_i\). More generally, such dependence may arise through some additional latent confounding variable \(L_i\).
Suppose that there exists a latent variable \(L_i\) such that, conditional on \((\mca,\cZ,L_i=\ell)\), the coordinates of \(\bG_i\) are independent with
\[
\bG_i\mid \mca,\cZ,L_i=\ell
\sim
\bigotimes_{j=1}^{n_i} Q_{ij}^{\ell},
\]
where \(Q_{ij}^{\ell} \in \cP([1,\Gamma])\).
Let \(\Pi_i\) denote the conditional distribution of \(L_i\) given
\((\mca,\cZ)\).  Then, 
\[
\P(\bG_i\in\cdot\mid \mca,\cZ)
=
\int
\left(\bigotimes_{j=1}^{n_i} Q_{ij}^{\ell}\right)(\cdot)
\,d\Pi_i(\ell).
\] 
Thus \(\P(\bG_i\in\cdot\mid \mca,\cZ)\) can be viewed as a mixture of product laws, with mixing weight distribution given by the conditional law \(\Pi_i\) of \(L_i\).
\end{example}
We now formalize such mixture-of-product representations by specifying \(\cP_i\) directly as the class of all mixtures over a base class of product laws.
Let \(\mcl_i^{\otimes}\) denote a \emph{base set} of product laws on \([1,\Gamma]^{n_i}\), where each
\(Q_i\in\mcl_i^{\otimes}\) has the form $Q_i=\bigotimes_{j=1}^{n_i} Q_{ij}$,
with \(Q_{ij} \in \cP([1,\Gamma])\).
For any probability measure \(\Lambda_i\) supported on \(\mcl_i^{\otimes}\), define the
induced mixture law \(P_{\Lambda_i}\) on \([1,\Gamma]^{n_i}\) by
\[
P_{\Lambda_i}(\cdot) = 
\int_{\mcl_i^{\otimes}} Q_i(\cdot)\,d\Lambda_i(Q_i).
\]
The class of all mixtures  over $\mcl_i^\otimes$ is denoted by
\[
\operatorname{Mix}(\mcl_i^{\otimes})
:=
\left\{
P_{\Lambda_i}:
\Lambda_i \text{ is a probability measure on }\mcl_i^{\otimes}
\right\}.
\]
It contains every product law in
\(\mcl_i^{\otimes}\), by taking \(\Lambda_i\) to be a point mass, but it also
contains non-product laws obtained by mixing over different product laws.
Note that any
\(P_{\Lambda_i}\in\operatorname{Mix}(\mcl_i^{\otimes})\) can be
generated through a two-stage sampling scheme. First, draw a random product law
\(Q_i\sim\Lambda_i\). Then, conditional on \( Q_i=\bigotimes_{j=1}^{n_i} Q_{ij}\), generate $\bG_i\mid\mca,\cZ, Q_i
\sim
\bigotimes_{j=1}^{n_i} Q_{ij}.$


In addition to flexibility and interpretability, the proposed mixture class also enjoys appealing theoretical properties. 
The following theorem shows that the separable algorithm can be implemented without optimizing directly over all mixtures, by reducing the problem to the corresponding one over the base product class.

\begin{theorem}\label{thm:mix-product-mean}
Let \(\mcl_i^{\otimes}\) be a base
class of product laws, and define $\cP_i
:=
\operatorname{Mix}(\mcl_i^{\otimes}).$
Then
\begin{equation}\label{eq:sup_same}
\sup_{P_i\in\cP_i}\mu_i(P_i)
=
\sup_{Q_i\in\mcl_i^{\otimes}}\mu_i(Q_i).
\end{equation}
Let $(\mcl_i^{\otimes})^\star
:=
\operatorname*{arg\,max}_{Q_i\in\mcl_i^{\otimes}}\mu_i(Q_i)$ and $\cP_i^\star
:=
\operatorname*{arg\,max}_{P_i\in\cP_i}\mu_i(P_i).$
Then $\cP_i^\star$ $= \{P_{\Lambda_i}:
\Lambda_i\bigl((\mcl_i^{\otimes})^\star\bigr)=1\}.$
Moreover, for every \(P_{\Lambda_i}\in\cP_i^\star\), $\nu_i^2(P_{\Lambda_i})
=
\int_{(\mcl_i^{\otimes})^\star}
\nu_i^2(Q_i)\,d\Lambda_i(Q_i).$
Consequently,
\[
\sup_{P_i\in\cP_i^\star}\nu_i^2(P_i)
=
\sup_{Q_i\in(\mcl_i^{\otimes})^\star}\nu_i^2(Q_i).
\]
In particular, if the supremum on the right is attained at \(Q_i^\star\in(\mcl_i^{\otimes})^\star\), then the degenerate mixture at $Q_i^\star$ attains the supremum on the left. 
\end{theorem}
Theorem~\ref{thm:mix-product-mean} implies that, to apply the separable algorithm over \(\cP_i\), it suffices to first maximize \(\mu_i(Q_i)\) over
the product laws \(Q_i\in\mcl_i^{\otimes}\), and then, among the product laws attaining this
maximum, choose the one with the largest \(\nu_i^2(Q_i)\).
In the remainder of this section, we describe three mixture classes with
different choices of the base product class and characterize the resulting least favorable distributions.
\subsection{Mean-band class}
If $\mcl_i^{\otimes}$ were unrestricted, it would include point masses on arbitrary configurations of \(\bG_i\). The corresponding optimization would then be maximized by a point mass on the extreme configuration in \eqref{eq:rosen_topk}, thereby reproducing the conventional sensitivity analysis. 
To move beyond this degenerate worst case, we consider a \emph{mean-band mixture class}. This class is generated by mixtures of product laws under which each coordinate has mean bounded away from the endpoint values
\(1\) and \(\Gamma\). 
Specifically, for any 
\(g\in[0,1/2]\), we define the set of laws:
\[
\mcl(g)
:=
\left\{
Q\in\cP([1,\Gamma]):
\mu^-(g)\le \E_{G\sim Q}[G] \le \mu^+(g)
\right\},
\]
where $\mu^-(g) = 1 + (\Gamma-1)g$ and $\mu^+(g) := \Gamma - (\Gamma-1)g$ denote the lower and upper bounds on the mean of $G\sim Q \in\cP([1,\Gamma])$. 
Therefore, the set $\mcl(g)$ imposes no parametric form: it constrains the laws only through their first moments.
To provide additional intuition for the parameter $g$, we consider the rescaled random variable $\widetilde G := (G-1)/(\Gamma-1) \in [0,1]$ with its distribution $\widetilde Q\in\cP([0,1])$. Under this transformation, the set $\mcl(g)$ can be equivalently written as
\[
\mcl(g)
:=
\left\{
\widetilde Q\in\cP([0,1]):
g\le \E_{\widetilde G\sim \widetilde Q}[\widetilde G] \le 1-g
\right\},
\] where $g$ and $1-g$ are the lower and upper bounds on the mean of $\widetilde G$.
Therefore, \(g\) controls the strength of the mean-band restriction: \(g=0\) imposes no restriction, while larger values of \(g\) force the means of each coordinate farther from the endpoints.
As a special case, when $\widetilde G$ follows a Bernoulli distribution, the mean-band constraint with $g > 0$ introduces randomness in $\widetilde G$ by requiring the success probability to be bounded away from 0 and 1; see Section \ref{subsec:ber} for further discussion on the Bernoulli mixture class.

Imposing this mean-band constraint coordinatewise gives the base product class:
\[
\mcl_i^{\otimes}(g)
:=
\left\{
Q_i = \bigotimes_{j=1}^{n_i}Q_{ij}:
Q_{ij}\in\mcl(g),\ j=1,\ldots,n_i
\right\}.
\]
The mixture class generated by this product class is
 $\cP_i(g)
:=
\operatorname{Mix}\!\left(\mcl_i^{\otimes}(g)\right).$
Although the mean-band restriction is imposed before mixing, every
\(P_i\in\cP_i(g)\) inherits the corresponding marginal mean-band restriction:
\begin{equation}\label{eq:mean-band}
  \mu^-(g)
  \le
  \E_{P_i}\bigl[G_{ij}\mid \mca,\cZ\bigr]
  \le
  \mu^+(g),
  \qquad j=1,\ldots,n_i.
\end{equation}
In this sense, the mixture class \(\cP_i(g)\) is distributionally robust: it permits
arbitrary mixtures over product laws while constraining the component marginal
laws only through their first moments, without imposing a parametric form.

By Theorem~\ref{thm:mix-product-mean}, the optimization over \(\cP_i(g)\) reduces to optimization over
\(\mcl_i^{\otimes}(g)\).
The following result shows that, despite the infinite-dimensional nature of
\(\cP_i(g)\), a least favorable distribution can be chosen to have a simple
discrete form.

\begin{proposition}
\label{prop:mean-band}
The supremum in
\begin{equation}\label{eq:obj_cpg}
\sup_{Q_i\in\mcl_i^{\otimes}(g)}
\mu_i(Q_i)
\end{equation}
is attained by $Q_i^\star =\bigotimes_{j=1}^{n_i}Q_{ij}^\star
\in \mcl_i^{\otimes}(g)$ such that each \(Q_{ij}^\star\) is supported on at most two points in
\([1,\Gamma]\). In particular, for each \(j\), either

\emph{(i)} \(Q_{ij}^\star=\delta_{c_{ij}}\) for some
\(c_{ij}\in[\mu^-(g),\mu^+(g)]\), where \(\delta_{c_{ij}}\) denotes the point
mass at \(c_{ij}\), or

\emph{(ii)} \(Q_{ij}^\star\) is supported on exactly two points in
\([1,\Gamma]\) and
\[
\E_{Q_{ij}^\star}[G_{ij}]
\in
\{\mu^-(g),\mu^+(g)\}.
\]
This $Q_{i}^\star$, viewed as a degenerate mixture, attains the worst-case mean over \(\cP_i(g)\).
\end{proposition}

Proposition~\ref{prop:mean-band} reduces the optimization in
\eqref{eq:obj_cpg} to a search over product laws whose marginal distributions
are supported on at most two points.  The problem can be reformulated as a nonconvex optimization program; see Section~\ref{app:nonlinear_pro} of the web-based supplementary  material for the reformulation and implementation details.

\subsection{Two-group class}
Although the optimization in \eqref{eq:obj_cpg} is numerically tractable for small matched-set sizes \(n_i\), it does not generally yield a simple characterization. 
To obtain a more explicit characterization, we next consider a \emph{two-group mixture class} \(\cP_i^{\mathrm{2G}}(g)\) under the mean-band constraint.  This class is generated by product laws with at
most two distinct marginal laws, both belonging to the set \(\mcl(g)\). That is, the two-group base product class is
\[
\mcl_i^{\otimes,\mathrm{2G}}(g)
:=
\left\{
Q_i=\bigotimes_{j=1}^{n_i}Q_{ij}:
\text{there exist } Q_i^+,Q_i^-\in\mcl(g)
\text{ such that }
Q_{ij}\in\{Q_i^+,Q_i^-\}
\right\}.
\]
The corresponding mixture class is $\cP_i^{\mathrm{2G}}(g)
:=
\operatorname{Mix}\!\left(\mcl_i^{\otimes,\mathrm{2G}}(g)\right).$
Since $\mcl_i^{\otimes,\mathrm{2G}}(g)
\subseteq
\mcl_i^{\otimes}(g),$
we also have $\cP_i^{\mathrm{2G}}(g)
\subseteq
\cP_i(g).$

The two-group product class is motivated by the structure of the worst-case
configuration in the conventional sensitivity analysis \eqref{eq:rosen_topk},
where units within a matched set are partitioned into two groups: one assigned
\(G_{ij}=\Gamma\), and the other assigned \(G_{ij}=1\). The present product class keeps this two-group structure before mixing, but allows the two groups to
have arbitrary marginal distributions \(Q_i^+,Q_i^-\in\mcl(g)\), rather than
being fixed at the deterministic point masses \(\delta_\Gamma\) and \(\delta_1\).
Notably, the two-group restriction is imposed \emph{before} mixing. An element of the
mixture class \(\cP_i^{\mathrm{2G}}(g)\) need not itself be a product law, nor need it have only two distinct marginal distributions after
marginalizing over the mixing distribution.
Because \(\cP_i^{\mathrm{2G}}(g)\subseteq\cP_i(g)\), every law in this subclass
also inherits the marginal mean-band restriction \eqref{eq:mean-band}.

By Theorem~\ref{thm:mix-product-mean}, the optimization over
\(\cP_i^{\mathrm{2G}}(g)\) reduces to the corresponding optimization over
\(\mcl_i^{\otimes,\mathrm{2G}}(g)\). The following theorem shows that this
reduced problem admits a finite top-\(k\) characterization.

\begin{theorem}\label{thm:two-group}
Suppose that \(q_{i1}\ge\cdots\ge q_{in_i}\).
For each \(k\in\{1,\ldots,n_i-1\}\), let \(Q_i^{(k)}\) denote the product law
with marginals
\[
Q_{ij}^{(k)}
=
\delta_{\mu^+(g)},
\text{ for } j\le k, \mbox{ and }
Q_{ij}^{(k)}
=
\operatorname{Bern}_{1,\Gamma}(g),
\text{ for } j>k.
\]
Then the optimization over \(\mcl_i^{\otimes,\mathrm{2G}}(g)\) reduces to a
finite search:
\begin{equation}\label{eq:obj_2G}
    \sup_{Q_i\in\mcl_i^{\otimes,\mathrm{2G}}(g)}
\mu_i(Q_i)
=
\max_{1\le k\le n_i-1}
\mu_i\!\left(Q_i^{(k)}\right).
\end{equation}
Let $\mck_i
:=
\operatorname*{arg\,max}_{1\le k\le n_i-1}
\mu_i\!\left(Q_i^{(k)}\right)$, and choose $k_i^\star
\in
\operatorname*{arg\,max}_{k\in\mck_i}
\nu_i^2\!\left(Q_i^{(k)}\right).$
Then $Q_i^{(k_i^\star)}$, viewed as the degenerate mixture, is a least favorable distribution over \(\cP_i^{\mathrm{2G}}(g)\).
\end{theorem}

Theorem~\ref{thm:two-group} reduces the optimization over the
two-group class to a finite search over \(k=1,\ldots,n_i-1\). The least
favorable distribution retains the classical top-\(k\) structure of the conventional
sensitivity analysis: the first \(k_i^\star\) units have a point mass marginal law
\(\delta_{\mu^+(g)}\), while, for the remaining \(n_i-k_i^\star\) units,
\(G_{ij}=1\) with probability \(1-g\) and \(G_{ij}=\Gamma\) with probability
\(g\). When \(g=0\), this becomes the deterministic  configuration in \eqref{eq:rosen_topk}. We refer to the resulting stochastic sensitivity
analysis as the
\emph{two-group analysis}.

\begin{corollary}
Suppose \(n_i=2\). Then $\mcl_i^{\otimes,\mathrm{2G}}(g)
=
\mcl_i^{\otimes}(g).$ Consequently, Theorem~\ref{thm:two-group} characterizes an optimizer of \eqref{eq:obj_cpg} for the full mean-band class in matched pairs.
\end{corollary}

\subsection{Bernoulli class}\label{subsec:ber}

We next consider the \emph{Bernoulli mixture class}
\(\cP_i^{\mathrm{Bern}}(g)\), which serves as another interpretable special case of the mean-band class. This class is generated by product laws whose marginals are endpoint Bernoulli distributions satisfying the
mean-band constraint. Specifically, the Bernoulli base product class is
\[
\mcl_i^{\otimes,\mathrm{Bern}}(g)
:=
\left\{
Q_i=\bigotimes_{j=1}^{n_i}\operatorname{Bern}_{1,\Gamma}(p_{ij}):
p_{ij}\in[g,1-g]
\right\}, 
\] 
where $\operatorname{Bern}_{1,\Gamma}(p)=
(1-p)\delta_1+p\delta_\Gamma$ denotes the endpoint Bernoulli distribution.
Unlike the two-group class, which restricts the number of distinct marginal
distributions within each product law, the Bernoulli product class
restricts each marginal distribution to be supported on the two endpoints
\(1\) and \(\Gamma\). The endpoint probabilities \(p_{ij}\), however, may vary
freely across \(j\), subject only to the mean-band constraint
\(p_{ij}\in[g,1-g]\). 
The corresponding mixture class is $\cP_i^{\mathrm{Bern}}(g)
:=
\operatorname{Mix}\!\left(\mcl_i^{\otimes,\mathrm{Bern}}(g)\right).$
Since $\mcl_i^{\otimes,\mathrm{Bern}}(g)
\subseteq
\mcl_i^{\otimes}(g),$
we also have $\cP_i^{\mathrm{Bern}}(g)
\subseteq
\cP_i(g).$

The Bernoulli product class is also motivated by the worst-case configuration in
the conventional sensitivity analysis \eqref{eq:rosen_topk}, where each
\(G_{ij}\) is fixed at one of the two endpoints, \(1\) or \(\Gamma\). The
present class keeps this endpoint structure, but replaces deterministic point
masses by 
\(\operatorname{Bern}_{1,\Gamma}(p_{ij})\) distributions, with
\(p_{ij}\in[g,1-g]\). When \(g=0\), the class contains the deterministic
endpoint configurations used in the conventional worst case. When \(g>0\),
such deterministic endpoint configurations are ruled out, because each endpoint
probability must be bounded away from both \(0\) and \(1\).
Again, the Bernoulli restriction is imposed before mixing. An element of the
mixture class \(\cP_i^{\mathrm{Bern}}(g)\) need not itself be a product
law; after marginalizing over the mixing distribution, the
coordinates may be dependent. Nevertheless, each coordinate remains supported
on \(\{1,\Gamma\}\), and each one-dimensional marginal remains an endpoint
Bernoulli distribution with parameter in \([g,1-g]\).

By Theorem~\ref{thm:mix-product-mean}, the optimization over
\(\cP_i^{\mathrm{Bern}}(g)\) reduces to the corresponding optimization over
\(\mcl_i^{\otimes,\mathrm{Bern}}(g)\). The following theorem shows that this
reduced problem again admits a finite top-\(k\) characterization.
\begin{theorem}\label{thm:ber}
Suppose that \(q_{i1}\ge\cdots\ge q_{in_i}\).
For each \(k\in\{1,\ldots,n_i-1\}\), let \(Q_i^{(k)}\) denote the
product law with marginals
\[
Q_{ij}^{(k)}
=
\operatorname{Bern}_{1,\Gamma}(1-g)
\quad\text{for } j\le k,\qquad Q_{ij}^{(k)}
=
\operatorname{Bern}_{1,\Gamma}(g)
\quad\text{for } j>k.
\]
Then the optimization over \(\mcl_i^{\otimes,\mathrm{Bern}}(g)\) reduces to a
finite search:
\[
\sup_{Q_i\in\mcl_i^{\otimes,\mathrm{Bern}}(g)}
\mu_i(Q_i)
=
\max_{1\le k\le n_i-1}
\mu_i\!\left(Q_i^{(k)}\right).
\]
Let $\mck_i
:=
\operatorname*{arg\,max}_{1\le k\le n_i-1}
\mu_i\!\left(Q_i^{(k)}\right).$
Choose $k_i^\star
\in
\operatorname*{arg\,max}_{k\in\mck_i}
\nu_i^2\!\left(Q_i^{(k)}\right).$
Then $Q_i^{(k_i^\star)}$, viewed as the degenerate mixture, is a least favorable distribution over \(\cP_i^{\mathrm{Bern}}(g)\).
\end{theorem}
Theorem~\ref{thm:ber} reduces the optimization over the
Bernoulli class to a finite search over \(k=1,\ldots,n_i-1\). The least favorable distribution again has a top-\(k\) structure: the first \(k_i^\star\) units take
the upper endpoint \(\Gamma\) with probability \(1-g\), while the remaining
units take \(\Gamma\) with probability \(g\).
When \(g=0\), this reduces to the deterministic configuration in
\eqref{eq:rosen_topk}. We refer to the
resulting stochastic sensitivity analysis as the \emph{Bernoulli analysis}.
To illustrate, consider matched pairs with \(n_i=2\) and \(q_{i1}\ge q_{i2}\). Then Theorem~\ref{thm:ber} yields \(p_{i1}^\star=1-g\) and \(p_{i2}^\star=g\), so $\P(G_{i1}=\Gamma)=1-g$ and $\P(G_{i2}=\Gamma)=g.$
Therefore,
\[
\P\!\left(\frac{G_{i1}}{G_{i2}}=\Gamma\right)=(1-g)^2,
\qquad
\P\!\left(\frac{G_{i1}}{G_{i2}}=1\right)=2g(1-g),
\qquad
\P\!\left(\frac{G_{i1}}{G_{i2}}=\Gamma^{-1}\right)=g^2.
\]
The worst-case configuration $G_{i1}=\Gamma\cdot G_{i2}$ occurs with probability $(1-g)^2$. At \(g=0\), this probability is one, recovering the
deterministic worst-case alignment.

In the motivating example of Section~\ref{subsec:mt_example}, one may write
\(G_{ij}=\Gamma^{U_{ij}}\), where \(U_{ij}=1\) indicates that subject \(j\) in
pair \(i\) carries the risk allele and \(U_{ij}=0\) otherwise. Then the Bernoulli parameter has the direct interpretation $p_{ij}=\P(U_{ij}=1).$
Thus \(g\le p_{ij}\le 1-g\) means that the probability of carrying the risk
allele is bounded away from both \(0\) and \(1\). 
Under the Bernoulli analysis, the least favorable distribution preserves the adversarial
alignment of the conventional sensitivity analysis, but makes that alignment
stochastic rather than deterministic.  In a discordant matched pair, for example, the subject who died of lung cancer carries the risk allele with probability \(1-g\), while the matched subject who did not die carries it with probability \(g\). 

Although \(\mcl_i^{\otimes,\mathrm{Bern}}(g)\) allows arbitrary
\(p_{ij}\in[g,1-g]\), Theorem~\ref{thm:ber} shows that a least favorable
product law uses only the two endpoint probabilities \(1-g\) and \(g\).
Thus, imposing a two-group restriction on the Bernoulli product class does
not change its worst-case mean.  Since the two-group
Bernoulli subclass is contained in \(\mcl_i^{\otimes,\mathrm{2G}}(g)\),
\begin{equation}\label{eq:2G_geq_ber}
    \sup_{Q_i\in \mcl_i^{\otimes,\mathrm{Bern}}(g)} \mu_i(Q_i) = \sup_{Q_i\in \mcl_i^{\otimes,\mathrm{Bern}}(g) \,\bigcap\, \mcl_i^{\otimes,\mathrm{2G}}(g)} \mu_i(Q_i) 
\le
\sup_{Q_i\in \mcl_i^{\otimes,\mathrm{2G}}(g)} \mu_i(Q_i).
\end{equation}
Thus, the two-group analysis yields a weakly more conservative worst-case mean
than the Bernoulli analysis. Section~\ref{sec:power} further compares the two classes.

\begin{remark}
The main developments that follow focus on the two-group and Bernoulli
analyses. Other parametric subclasses are also possible. As an illustration,
Section~\ref{subsec:stoch-beta} of the web-based supplementary material considers a
rescaled Beta subclass and shows that, despite its smooth parametric form, it
can approximate the supremal value in \eqref{eq:obj_2G} arbitrarily closely.
\end{remark}
\section{How stochasticity modifies the robustness frontier}\label{sec:power}
\subsection{Overview and setup}\label{subsec:design_setup}
The preceding sections develop stochastic sensitivity analyses indexed by the stochastic parameter \(g\). We now examine how the conclusions of the sensitivity analysis change as \(g\) moves from zero to small positive values. 
Because positive \(g\) imposes additional stochasticity restrictions on the distributions for hidden confounders, one should expect rejection of the null hypothesis to persist for larger values of \(\Gamma\).
Here we investigate the size of this increase when only a small amount of stochasticity is imposed.
To study this, we consider a \emph{favorable} setting in which a genuine treatment effect is present and there is no unmeasured confounding \citep{Rosenbaum01062010}, so that treatment assignment within each matched set is completely randomized.
The sensitivity analyses are nevertheless evaluated at sensitivity levels \(\Gamma\ge 1\), as if hidden bias of strength \(\Gamma\) were possible. 
In this favorable setting, the probability of rejection measures whether the sensitivity analysis can still detect the treatment effect after allowing for hidden bias of strength \(\Gamma\). 
We study this from two perspectives. 
First, we compare the conventional sensitivity analysis with the two-group and the Bernoulli analyses developed in Section~\ref{sec:classes} through the notion of \emph{design sensitivity} \citep{ros04}. 
Second, we fix \(\Gamma\) and study the relaxation parameter \(g\), asking how much departure from deterministic worst-case alignment is needed for rejection to persist.

Throughout this section we restrict attention to matched pairs, so that
$n_i=2$. To obtain explicit numerical comparisons, we consider a simple paired-difference generative
model. Let $D_i$ denote the treated-minus-control difference in the $i$-th pair and $\bar q_i=(q_{i1}+q_{i2})/{2}$, then 
\begin{equation}\label{eq:gen_model1}
        T_i-\bar q_i=\frac{D_i}{2},
    \qquad
    |T_i-\bar q_i|=\frac{|D_i|}{2}.
\end{equation}
In the favorable situation where there is an effect $\tau >0$ and the bias is absent, suppose $D_i=\tau+\varepsilon_i$,
where $\varepsilon_i$ are i.i.d. from some distribution to be specified below. 
\subsection{Design sensitivity}
\label{subsec:sim_design}

\citet{ros04} introduced design sensitivity to assess the limiting performance
of test statistics in the conventional sensitivity analysis. We use the same concept
to compare the conventional sensitivity analysis with the stochastic sensitivity
analyses.

Under the generative model described in the previous subsection, the design
sensitivity \(\widetilde\Gamma\) of a sensitivity analysis is the value of $\Gamma$
with the following property: for
\(\Gamma<\widetilde\Gamma\), the sensitivity analysis rejects Fisher's sharp
null with probability tending to one, whereas for
\(\Gamma>\widetilde\Gamma\), the probability of rejection tends to zero.
That is, \(\widetilde{\Gamma}\) summarizes how much hidden bias a study can sustain in the asymptotic sense.
Larger values of \(\widetilde\Gamma\) therefore indicate greater insensitivity to hidden bias. Under the generative model above, \(\widetilde{\Gamma}\) depends on the effect size \(\tau\).
For the stochastic analysis, the design sensitivity also depends on the relaxation parameter \(g\), and the choice of mixture class being optimized over. We write $\widetilde\Gamma(\tau;g)$ explicitly for the dependence, whereas the conventional sensitivity analysis corresponds to \(g=0\). Closed-form expressions for these design sensitivities under the
generative model above are given in Section~\ref{app:ds} of the web-based
supporting material.

\begin{figure}[t]
\centering
\begin{subfigure}{0.49\textwidth}
  \centering
  \includegraphics[width=\linewidth]{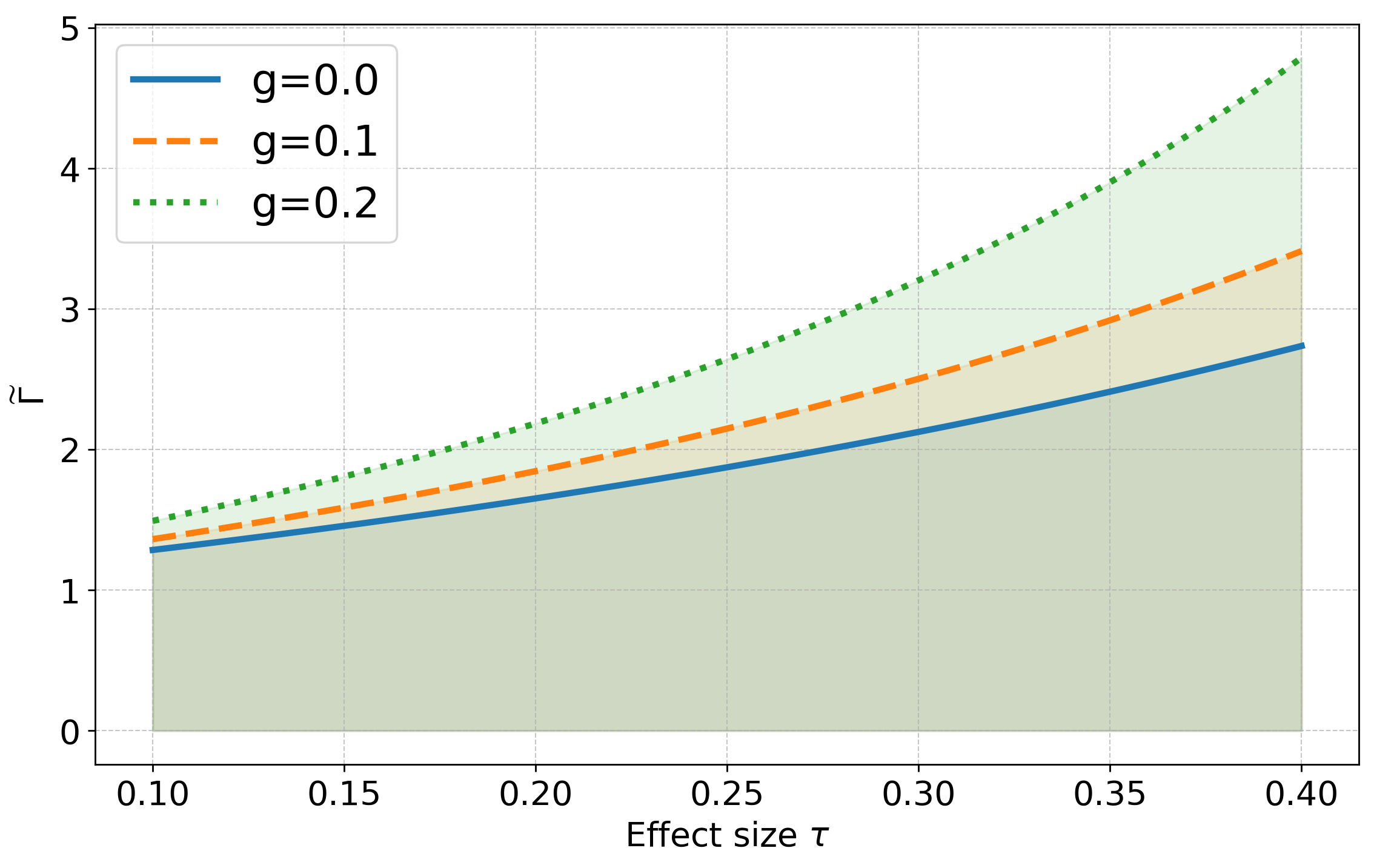}
  \caption{Gaussian noise (two-group)}
\end{subfigure}\hfill
\begin{subfigure}{0.49\textwidth}
  \centering
  \includegraphics[width=\linewidth]{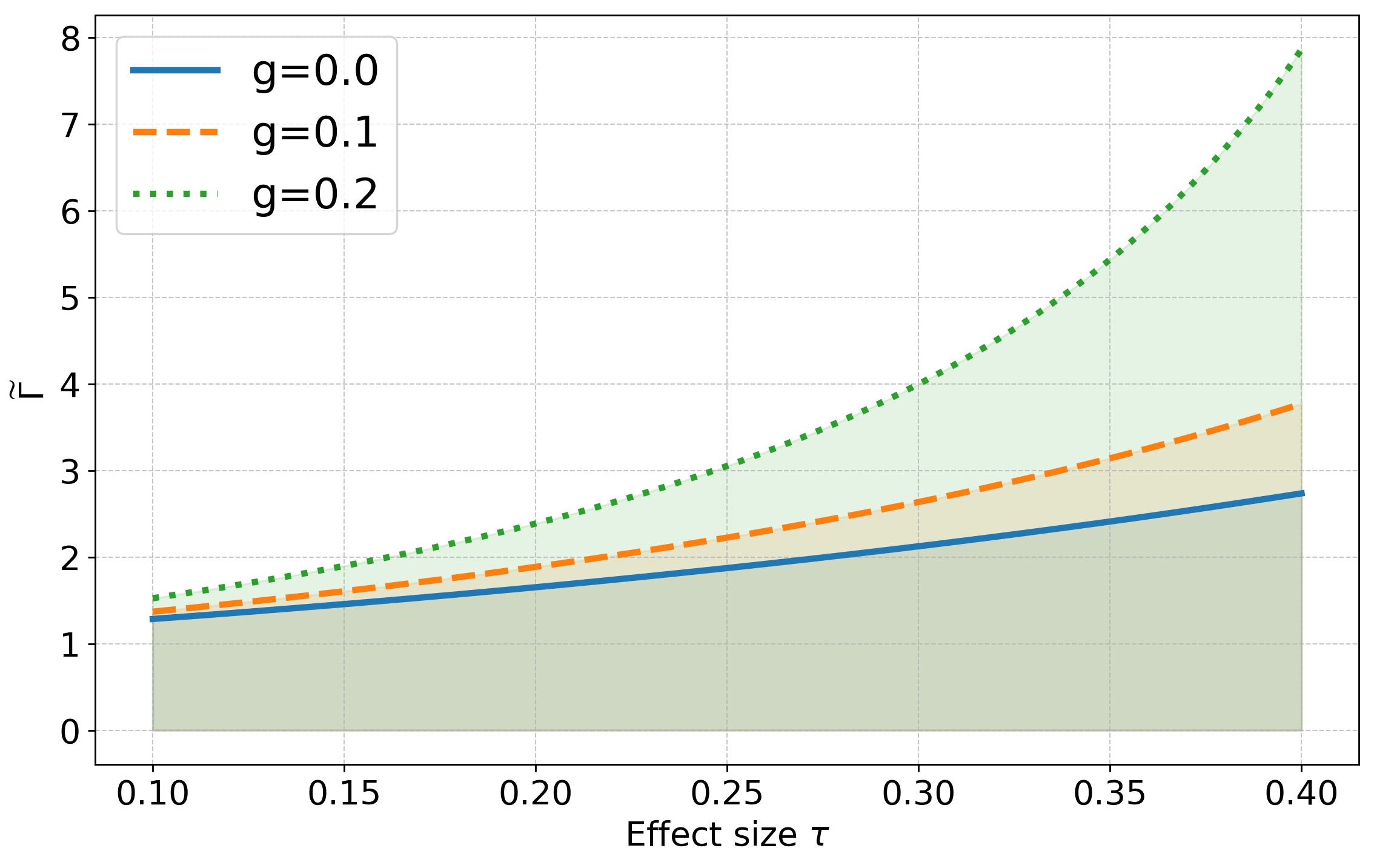}
  \caption{Gaussian noise (Bernoulli)}
\end{subfigure}
\medskip
\begin{subfigure}{0.49\textwidth}
  \centering
  \includegraphics[width=\linewidth]{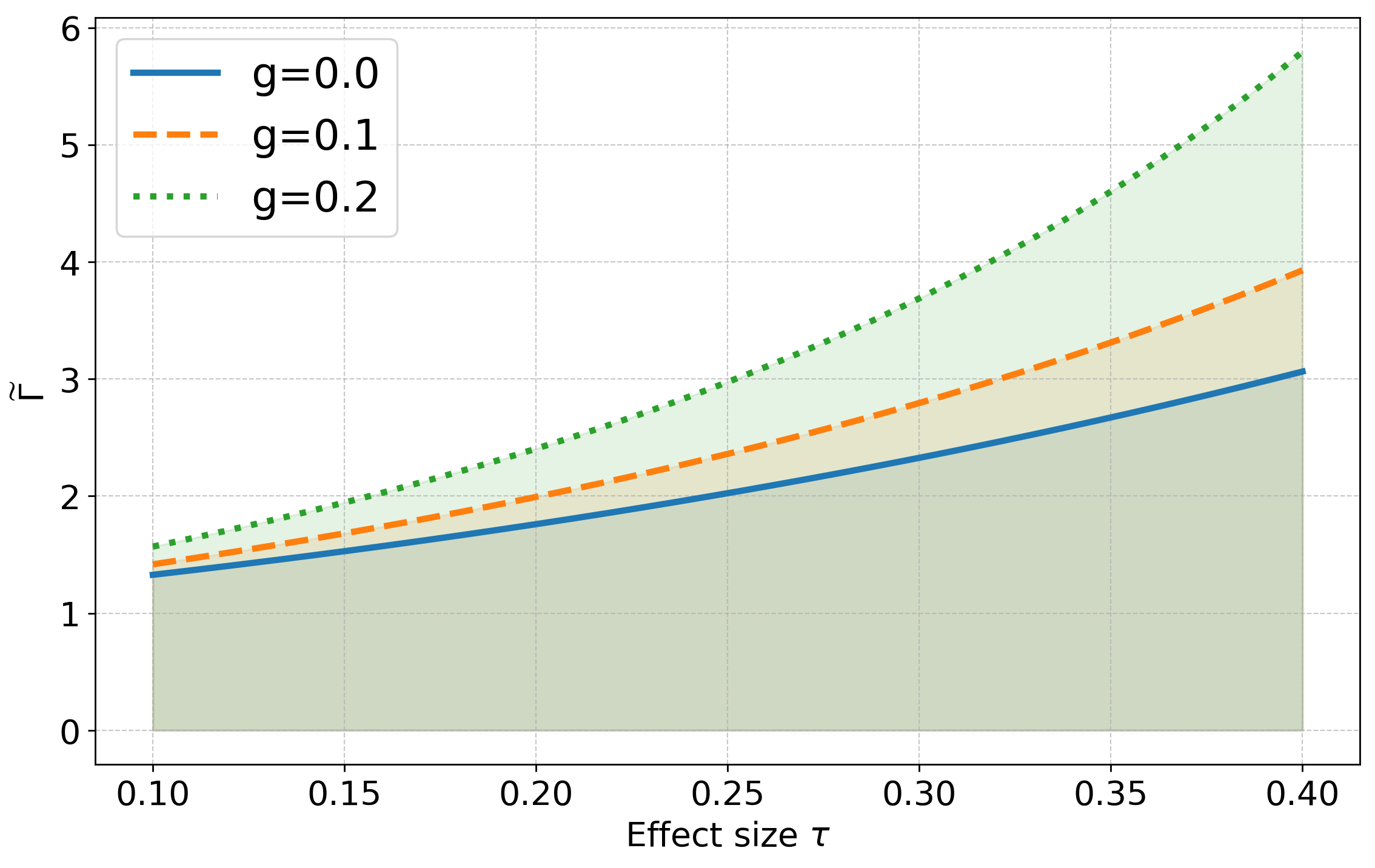}
  \caption{\(t_4\) noise (two-group)}
\end{subfigure}\hfill
\begin{subfigure}{0.49\textwidth}
  \centering
  \includegraphics[width=\linewidth]{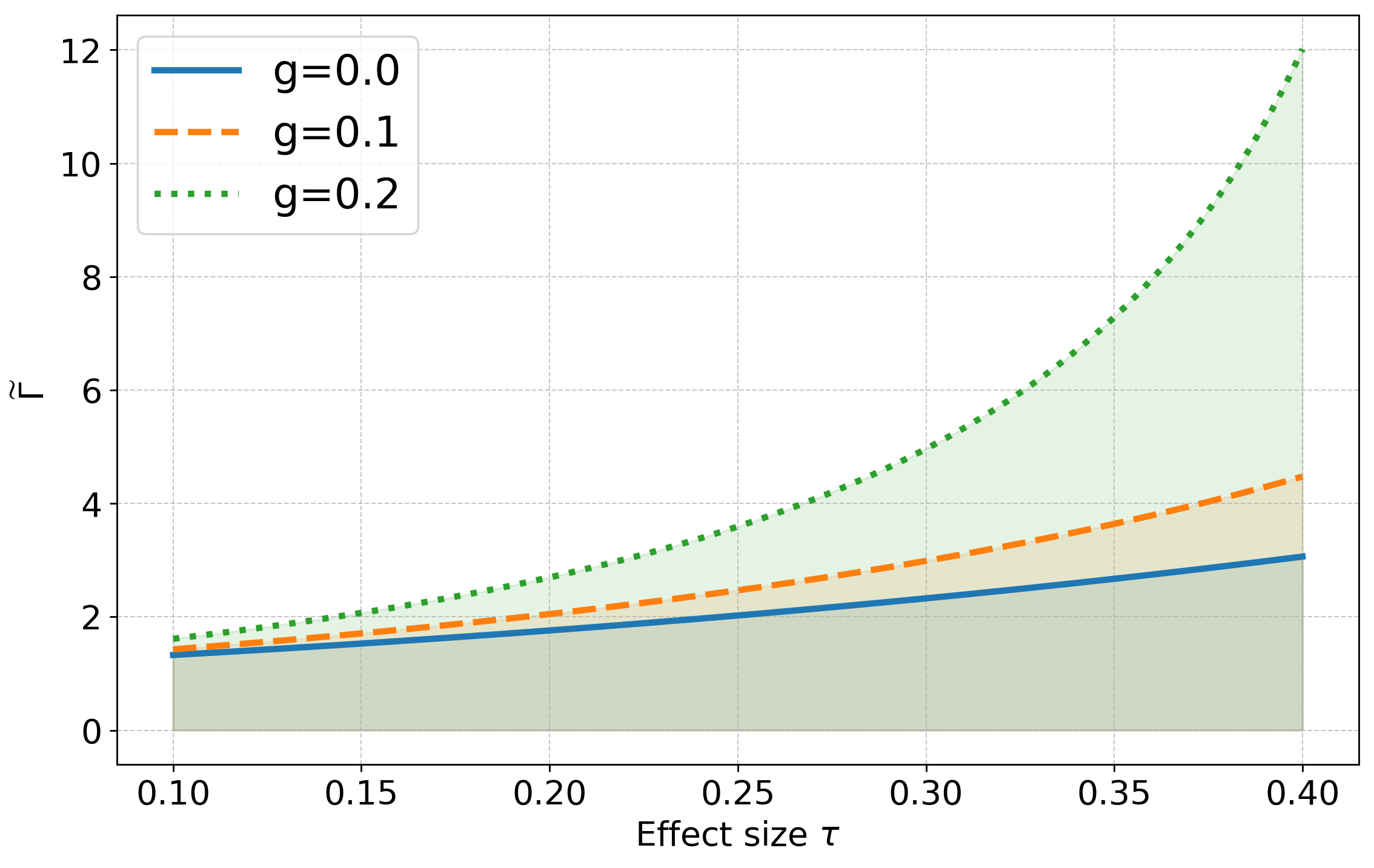}
  \caption{\(t_4\) noise (Bernoulli)}
\end{subfigure}

\caption{Design sensitivity \(\tilde{\Gamma}(\tau;g)\) for different noise distributions, classes of conditional laws for the unobserved confounder, and degrees of stochasticity. Shaded regions indicate rejection of Fisher's sharp null.}
\label{fig:ds}
\end{figure}

Figure~\ref{fig:ds} plots \(\widetilde{\Gamma}(\tau;g)\) for the two stochastic sensitivity analyses, together with the conventional sensitivity analysis, over \(\tau\in(0,0.4]\), \(g\in\{0, 0.1,0.2\}\), and two distributions for \(\varepsilon_i\): Gaussian noise and \(t_4\) noise, both with mean \(0\) and variance \(1\). Across all panels, design sensitivity increases with \(\tau\), reflecting the fact that stronger treatment effects can withstand larger hidden bias. 
For every fixed \(\tau\), the curve for \(g=0.2\) lies above the curve for \(g=0.1\), and both lie above the curve for \(g=0\). This is expected: increasing \(g\) tightens the mean-band restriction and therefore narrows the candidate class of distributions for the hidden confounders. The resulting worst case is less adversarial, so rejection persists for larger values of \(\Gamma\). More revealing within Figure~\ref{fig:ds} is the magnitude of this change: even small positive values of \(g\) can substantially increase the design sensitivity. For example, under Gaussian noise with \(\tau=0.25\), the conventional sensitivity analysis has design sensitivity \(\widetilde\Gamma=1.87,\) meaning that the probability of rejecting the sharp null tends to zero beyond this level. In contrast, when \(g=0.2\), the two-group analysis has design sensitivity \(\widetilde\Gamma=2.64\), while the Bernoulli analysis has design sensitivity \(\widetilde\Gamma=3.05\). 
Thus, in this setting, imposing \(g=0.2\) changes the robustness statement: relative to \(g=0\) of the conventional sensitivity analysis, rejection persists to larger values of \(\Gamma\). Additionally, comparing the two stochastic analyses, the Bernoulli analysis yields larger design sensitivities than the two-group analysis, consistent with \eqref{eq:2G_geq_ber}: for the same value of \(g\), the Bernoulli analysis uses a smaller least favorable expectation and hence gives a less adversarial sensitivity analysis.

\subsection{How much stochasticity is needed?}
In the previous subsection, for each stochastic sensitivity analysis, we  fix \(g\) and ask how large \(\Gamma\) can be before we fail to reject. We now take the complementary view: for a fixed value of \(\Gamma\), how large must \(g\) be for rejection to persist?

For a stochastic sensitivity analysis and a fixed sensitivity parameter \(\Gamma\), let $\tilde{g}$ be the value of $g$ with the following property: for
\(g>\tilde{g}\), the sensitivity analysis rejects Fisher's sharp
null with probability tending to one, whereas for
\(g<\tilde{g}\), the probability of rejection tends to zero.
Thus, \(\tilde{g}\) is the smallest stochasticity restriction under which the stochastic analysis can sustain hidden bias of size \(\Gamma\).
If \(\tilde{g}=0\), then the conventional sensitivity analysis already detects the treatment effect at that value of \(\Gamma\). If \(\tilde{g}>0\), then the conventional sensitivity analysis fails to reject, whereas the corresponding stochastic sensitivity analysis can still reject provided \(g>\tilde{g}\). Under the above generative model, $\tilde{g}$ depends on the effect size $\tau$. We therefore write \(\tilde{g}(\tau;\Gamma)\) to make explicit the dependence.

\begin{figure}[t]
\centering
\begin{subfigure}{0.49\textwidth}
  \centering
  \includegraphics[width=\linewidth]{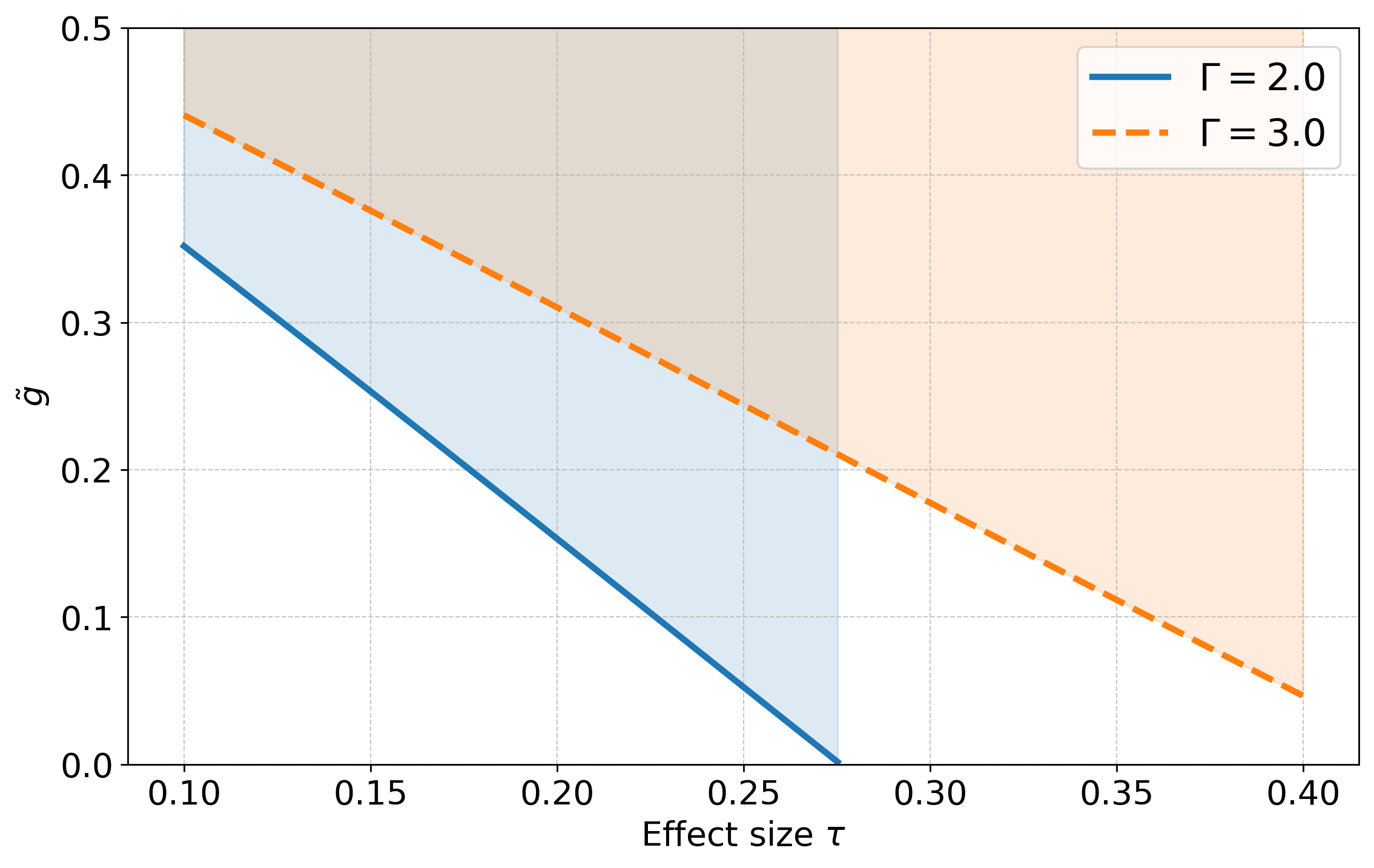}
  \caption{Gaussian noise (two-group)}
\end{subfigure}\hfill
\begin{subfigure}{0.49\textwidth}
  \centering
  \includegraphics[width=\linewidth]{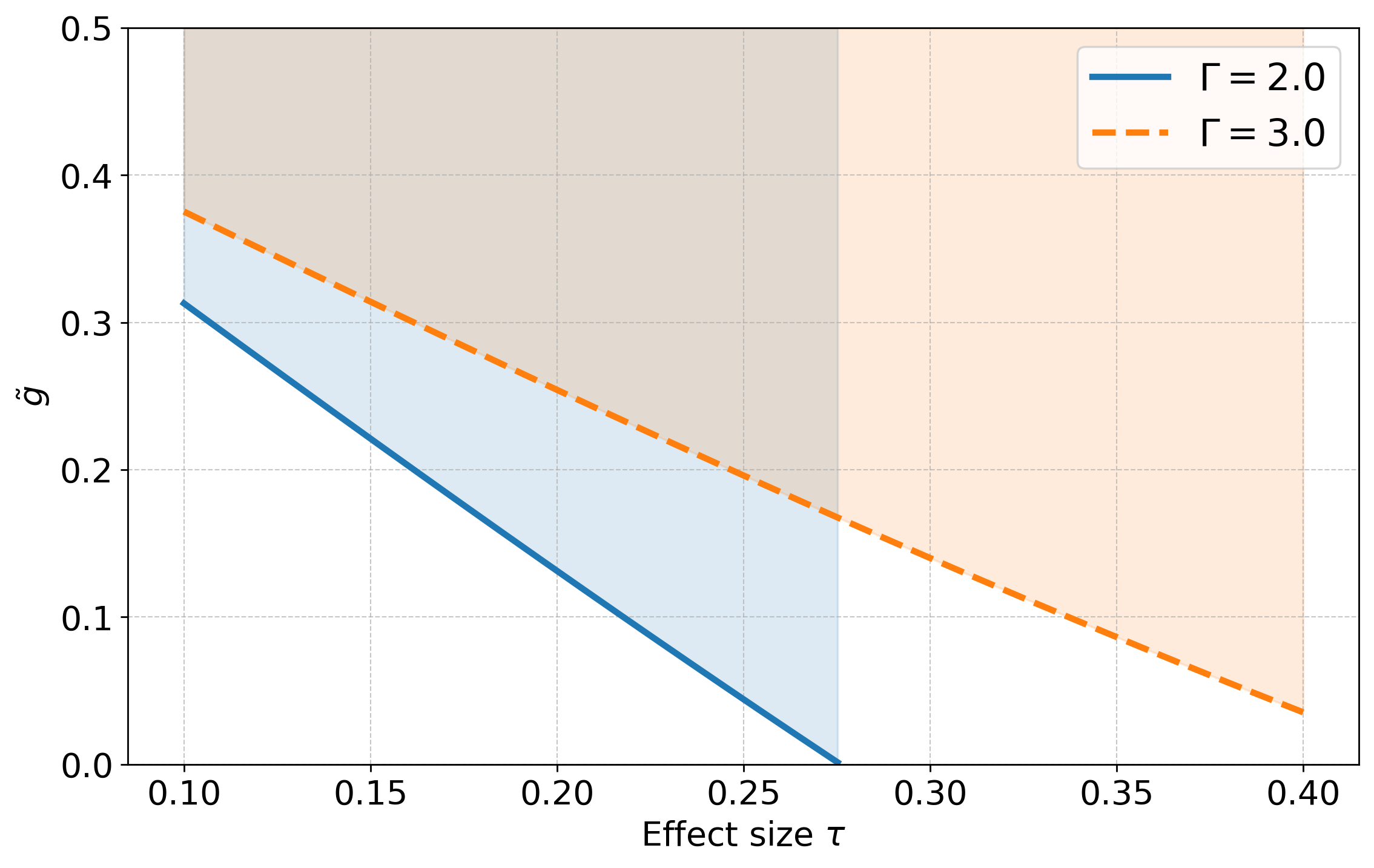}
  \caption{Gaussian noise (Bernoulli)}
\end{subfigure}

\medskip

\begin{subfigure}{0.49\textwidth}
  \centering
  \includegraphics[width=\linewidth]{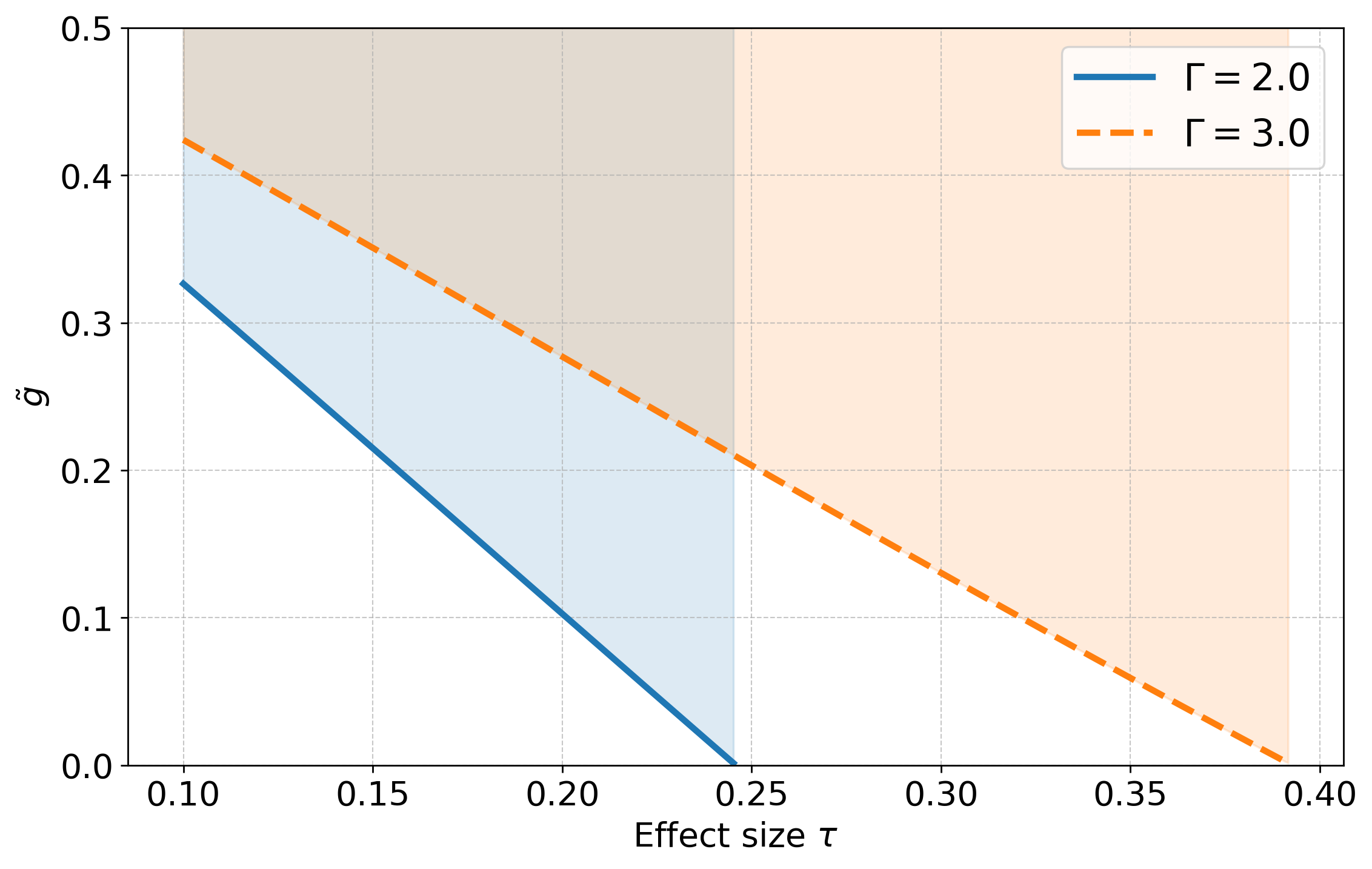}
  \caption{\(t_4\) noise (two-group)}
\end{subfigure}\hfill
\begin{subfigure}{0.49\textwidth}
  \centering
  \includegraphics[width=\linewidth]{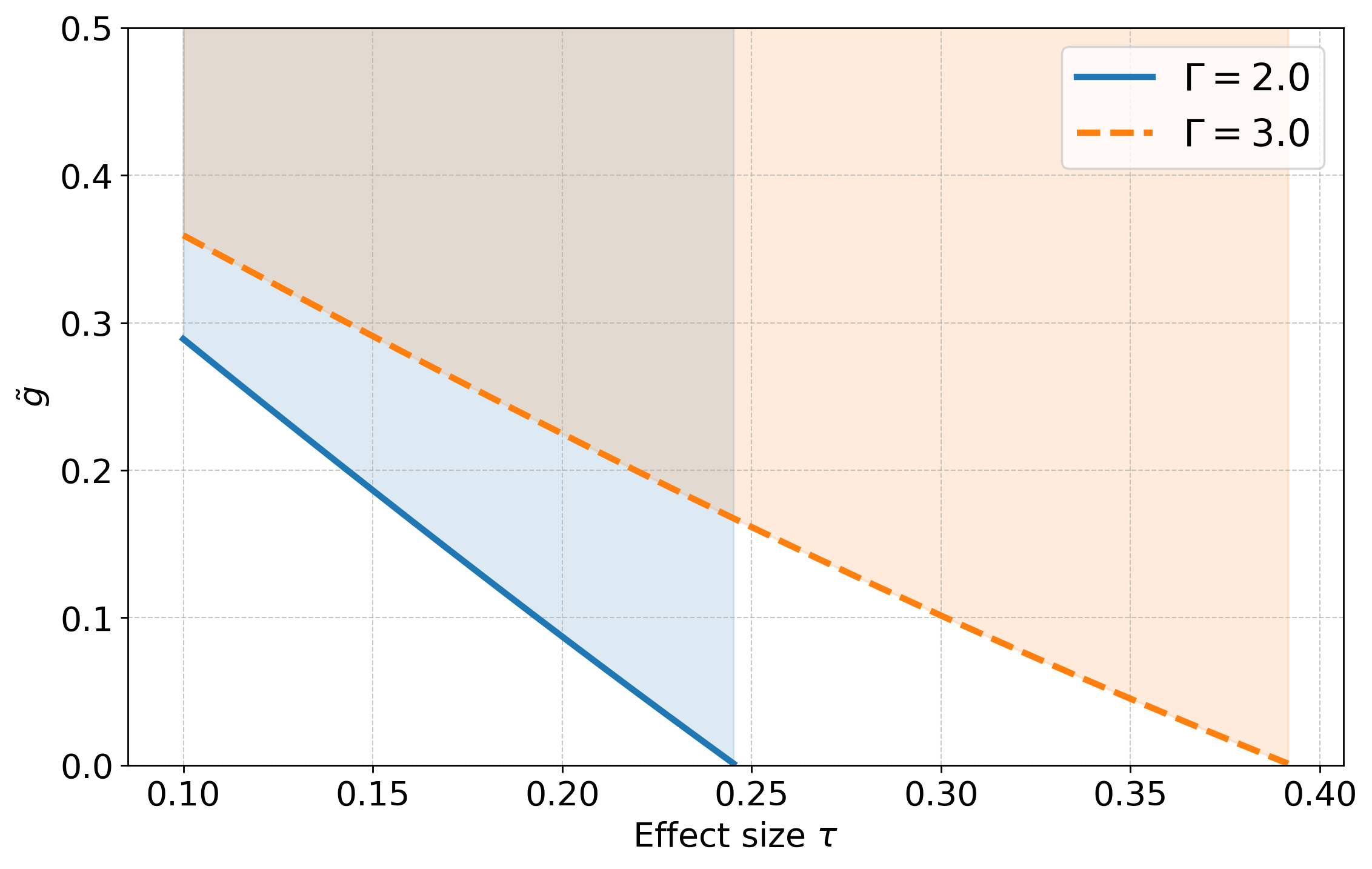}
  \caption{\(t_4\) noise (Bernoulli)}
\end{subfigure}

\caption{\(\tilde{g}(\tau;\Gamma)\) for different noise distributions, classes of conditional laws for the unobserved confounder, and degrees of stochasticity.  Shaded regions indicate rejection of Fisher's sharp null.}
\label{fig:gcrit}
\end{figure}

Figure~\ref{fig:gcrit} plots \(\tilde{g}(\tau;\Gamma)\) for the two stochastic sensitivity analyses, over \(\tau\in(0.1,0.4]\), \(\Gamma\in\{2, 3\}\), and two distributions for \(\varepsilon_i\): Gaussian noise and \(t_4\) noise, both with mean \(0\) and variance \(1\). 
Across both stochastic sensitivity analyses, the threshold \(\tilde{g}(\tau;\Gamma)\) decreases as the effect size \(\tau\) increases, showing that stronger treatment effects require less stochastic relaxation of the conventional sensitivity analysis in order to be detected.
For example, under the Gaussian noise and \(\Gamma=2\), \(\tilde{g}\) for two-group analysis decreases from about \(0.35\) at \(\tau=0.1\) to about \(0.05\) at \(\tau=0.25\), and reaches zero at approximately \(\tau=0.28\). 
Thus, when the effect size is modest, rejection may require only a small positive value of \(g\), whereas the conventional sensitivity analysis rejects only when the effect size is sufficiently large.
The threshold is generally smaller for the Bernoulli analysis than for the two-group analysis, again reflecting that the Bernoulli analysis is less conservative; see \eqref{eq:2G_geq_ber}. The exception occurs when
\(\tilde g(\tau;\Gamma)=0\), in which case the two analyses necessarily agree,
because at \(g=0\) both reduce to the conventional sensitivity analysis. For instance, under the effect size $\tau = 0.28$ with Gaussian noise and \(\Gamma=2\), both analyses reach \(\tilde g(\tau;\Gamma)=0\).

\section{Data illustrations}\label{sec:application}
\subsection{Reanalysis of Hammond's smoking study}\label{subsec:hammond_app}
In the smoking study of \citet{ham64}, introduced in Section~\ref{subsec:mt_example}, there were 122 discordant matched pairs in which exactly one subject died of lung cancer. Among these discordant pairs, 110 deaths occurred among smokers and 12 occurred among nonsmokers.  We reanalyze these data using McNemar's test statistic under the two-group analysis and the Bernoulli analysis developed in Section~\ref{sec:classes}. For each relaxation parameter \(g\in[0,0.1]\), we compute the corresponding sensitivity value \(\hat{\Gamma}\) \citep{zhao19}, defined as the smallest value of \(\Gamma\) at which the given sensitivity analysis no longer rejects Fisher's sharp null at level \(\alpha=0.05\).

\begin{figure}[htbp]
\centering
\includegraphics[width=0.5\textwidth]{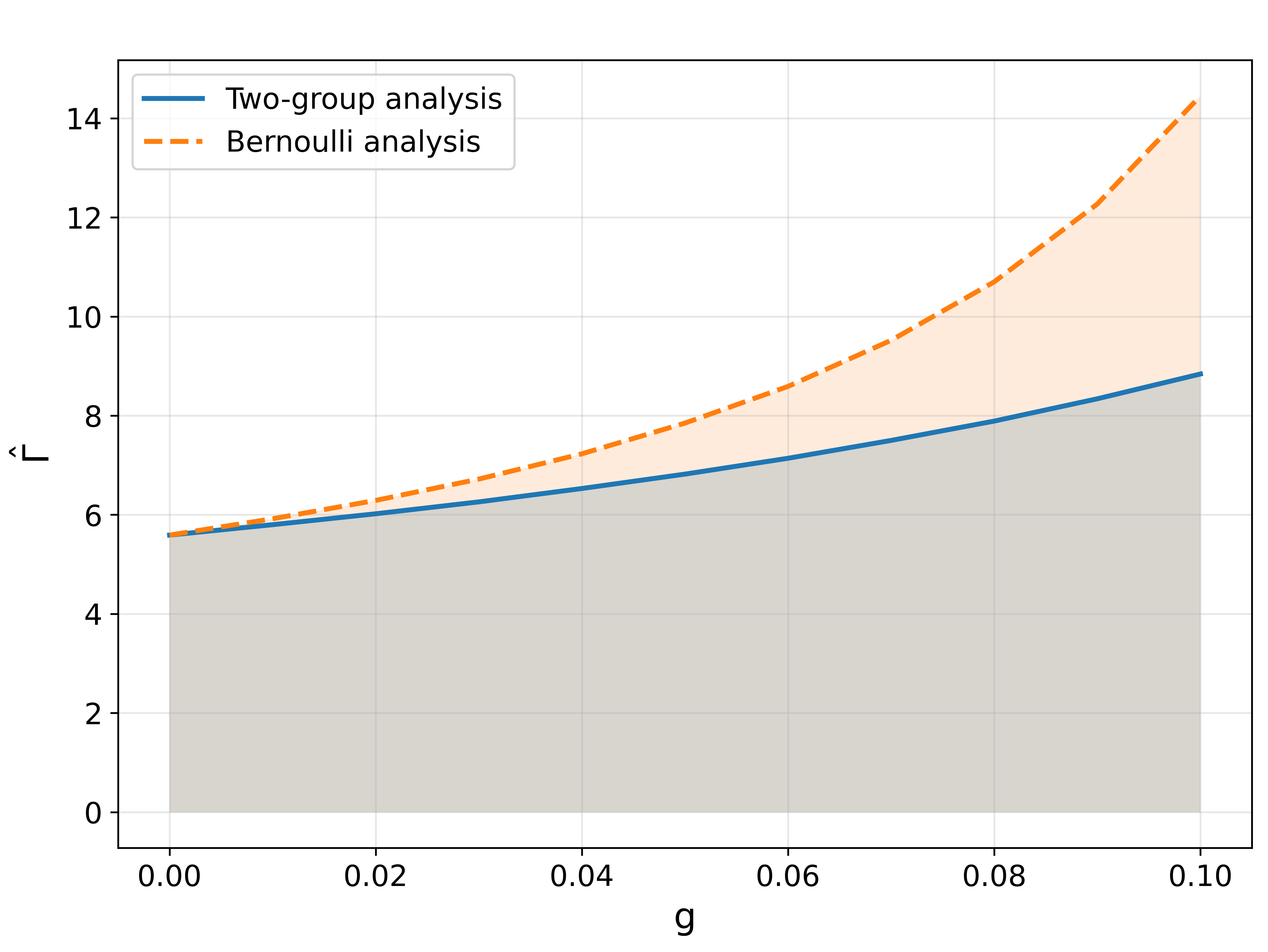}
\caption{Sensitivity values for different classes of conditional laws for the unobserved confounder, and degrees of stochasticity. Shaded regions indicate rejection of Fisher’s sharp null.}
\label{fig:hammond_svals}
\end{figure}

Figure~\ref{fig:hammond_svals} shows the results. Hammond's study is already highly insensitive under the conventional sensitivity analysis: at \(g=0\), both stochastic analyses coincide with the conventional sensitivity analysis and yield a sensitivity value of \(5.59\). Yet, by allowing only a modest departure from deterministic worst-case alignment, the gains in robustness are substantial.
Under the two-group analysis, the sensitivity value increases monotonically from $5.59$ at $g=0$ to $8.84$ at $g=0.1$. Under the Bernoulli analysis, the increase is much larger, reaching $14.45$ at $g=0.1$. 

The Bernoulli analysis is particularly interpretable in this example. As discussed in Section~\ref{subsec:mt_example}, a plausible hidden confounder is whether a subject carries a risk allele at the relevant variant. Such an allele may affect both smoking behavior and lung-cancer risk, but the available scientific evidence does not suggest perfect alignment between carrier status and death. Under the least favorable conditional law in the Bernoulli analysis, in a discordant pair, the subject who died of lung cancer carries the risk allele with probability \(1-g\), while the matched subject carries it with probability \(g\). Thus, when \(g=0.1\), the Bernoulli analysis still permits a highly adverse stochastic confounder, under which the subject who died of lung cancer carries the risk allele with probability \(0.9\) and the matched subject with probability \(0.1\). Even under such a strongly adverse specification, the sensitivity value increases from \(5.59\) to \(14.45\). 
Thus, under this Bernoulli interpretation, the observed association is difficult to attribute solely to a stochastic hidden confounder of the form considered here.

\subsection{Reanalysis of Binge drinking study}\label{subsec:binge}
We illustrate our methods with a reanalysis of the NHANES binge drinking study in \citet{ros23evidence}, based on the 2017--2020 National Health and Nutrition Examination Survey (NHANES) \citep{Akinbami2022}. The study was motivated by prior evidence that heavy episodic alcohol consumption may increase blood pressure \citep{Roerecke2017}. Unlike \citet{ros23evidence}, which used two control groups, we focus on the comparison of frequent binge drinkers and never-bingers and therefore use a different matching design.

Following \citet{ros23evidence}, we compare frequent binge drinkers to never-bingers on three outcomes: systolic blood pressure (SBP), diastolic blood pressure (DBP), and a pre-specified weighted combination: $(\mathrm{DBP}/10.7)+(\mathrm{SBP}/14.7).$
To reduce bias from measured covariates, we form matched sets using variable-ratio matching implemented through restricted full matching \citep{han04} on nine pre-treatment baseline covariates.  The resulting matched sample contains 1,382 individuals, including 206 treated units, with each treated unit matched to up to ten controls. Details on the
covariates and matching procedure are provided in
Section~\ref{app:binge-match} of the web-based supporting material.
As shown in Figure~\ref{fig:balance}, covariate balance improved substantially: large pre-matching standardized differences were reduced to near zero after matching, with all covariates falling well within the conventional balance threshold.
We test Fisher’s sharp null using a Huber M-score sum statistic \citep{hub81} with outer trimming constant 2.5.

\begin{figure}[t]
  \centering
  \includegraphics[width=0.8\linewidth]{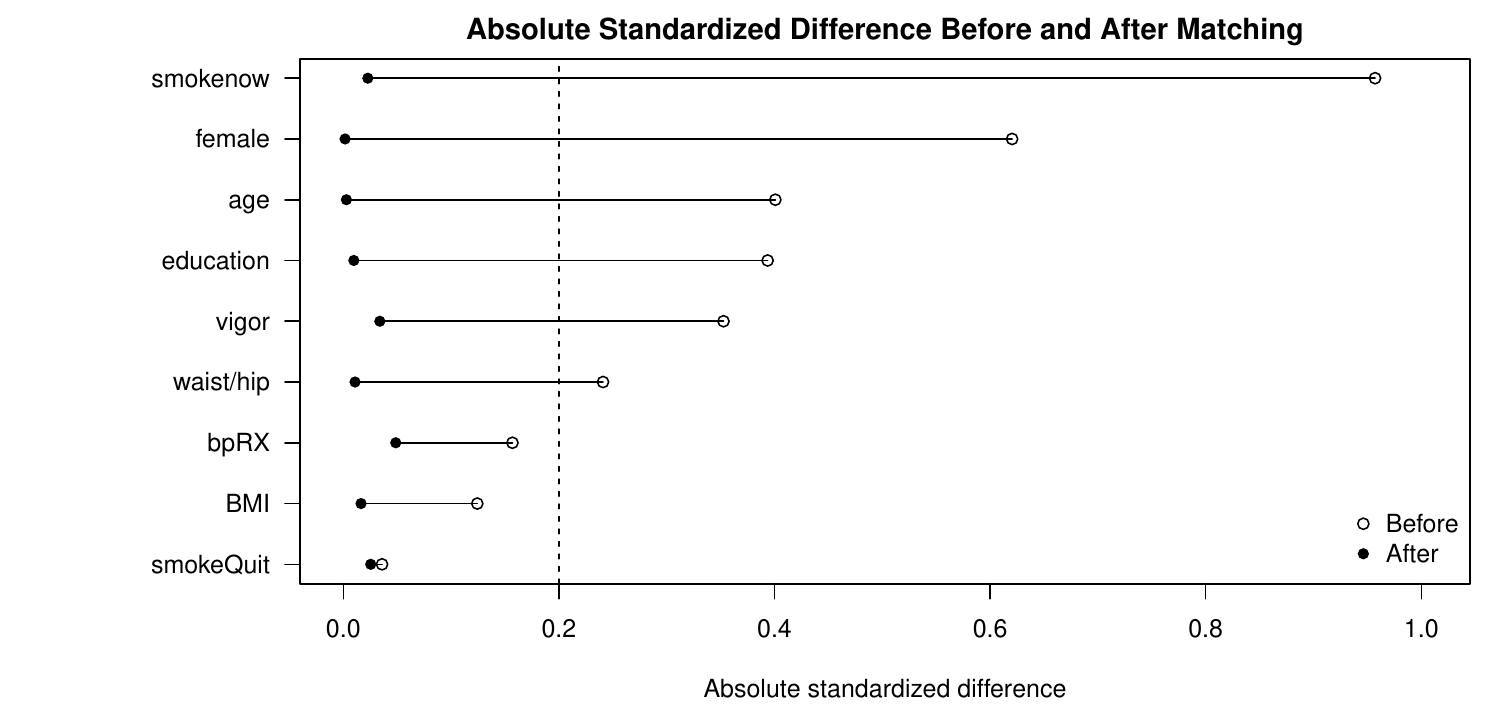}
  \caption{Covariate imbalances before and after matching. The vertical reference line indicates a threshold of 0.2, which is often regarded as the maximal allowable absolute standardized difference \citep{sil01}.}
  \label{fig:balance}
\end{figure}
We compare three sensitivity analyses: the conventional sensitivity analysis, the two-group analysis and the Bernoulli analysis using \(g \in \{0.1, 0.2\}\). For each outcome and each analysis, we report the corresponding sensitivity value \(\widehat{\Gamma}\).
Table~\ref{tab:sensvals_one_table} shows how the sensitivity value changes when
one moves from the conventional sensitivity analysis at \(g=0\) to the
stochastic analyses with positive \(g\). For the weighted-combination outcome, for example, the conventional sensitivity analysis ceases to reject at \(\widehat{\Gamma}=2.37\), whereas the two-group analysis continues to reject up to \(\widehat{\Gamma}=2.94\) when setting \(g=0.1\) and \(\widehat{\Gamma}=4.27\) when setting \(g=0.2\). The corresponding sensitivity values under the Bernoulli analysis are \(\widehat{\Gamma}=3.02\) and \(\widehat{\Gamma}=4.86\).
\begin{table}[h]
\centering
\begin{tabular}{lccccc}
\toprule
& \multicolumn{1}{c}{Conventional}
& \multicolumn{2}{c}{Two-group}
& \multicolumn{2}{c}{Bernoulli} \\
\cmidrule(lr){2-2}\cmidrule(lr){3-4}\cmidrule(lr){5-6}
Outcome & $g = 0$ & $g=0.1$ & $g=0.2$ & $g=0.1$ & $g=0.2$ \\
\midrule
Systolic blood pressure (SBP)  & 2.10 & 2.52 & 3.42 & 2.57 & 3.74 \\
Diastolic blood pressure (DBP) & 2.26 & 2.77 & 3.91 & 2.84 & 4.38 \\
Weighted combination           & 2.37 & 2.94 & 4.27 & 3.02 & 4.86 \\
\bottomrule
\end{tabular}
\caption{Sensitivity values \(\widehat{\Gamma}\) for three outcomes under three types of sensitivity  analyses.}
\label{tab:sensvals_one_table}
\end{table}

As discussed in Section~\ref{sec:classes}, as \(g\) increases, the analyses rule out progressively more extreme confounder distributions, are therefore less conservative, leading to larger sensitivity values. Moreover, Equation~\eqref{eq:2G_geq_ber} implies that the two-group analysis is weakly more conservative than the Bernoulli analysis, so it is unsurprising that the Bernoulli analysis yields slightly larger sensitivity values in Table~\ref{tab:sensvals_one_table}. 
Overall, the application shows that a modest relaxation of deterministic worst-case alignment can materially change the robustness conclusion.

\section{Discussion}
This paper proposes a stochastic sensitivity analysis for matched observational
studies that complements the conventional sensitivity analysis by
adding a second parameter \(g\), which restricts how concentrated the hidden
confounder distribution may be near its most adverse configurations. The framework separates two aspects of hidden bias: \(\Gamma\) controls the magnitude of bias, while \(g\) controls the extent to which the hidden confounder may align with the most adverse configuration.


Our results focus on the two-group and Bernoulli analyses, which impose additional structure on the mean-band class. The mean-band class \(\mcl_i^{\otimes}(g)\), by contrast, imposes only marginal mean restrictions and allows the coordinates within a matched set to have distinct marginal distributions. Proposition~\ref{prop:mean-band} shows that the optimization over this class is not intractable: an optimizer can always be chosen with at most two support points for each
marginal distribution. What remains open is to characterize these two-point
optimizers more explicitly. A natural next step is therefore to determine when
the two-group solution is already optimal for the full mean-band class, when
more general two-point marginals are needed, and how to compute the resulting
optimizer efficiently for larger matched sets.

It would also be natural to allow the stochasticity parameter to vary across
matched sets. In the current formulation, \(g\) is global: every matched set is
subject to the same restriction on how close the confounder distribution may come to the deterministic worst-case alignment. This specification is
convenient, but real studies may contain matched sets with different degrees of
adverse confounding. In some sets, the hidden confounder may behave nearly like the deterministic worst-case configuration considered by the conventional sensitivity analysis; in others, the hidden confounder may vary more stochastically and be less tightly aligned with the most adverse configuration.
A heterogeneous extension could replace the scalar \(g\) by set-specific parameters \(g_i\), or by a mixture structure in which some fraction of matched sets
remain unrestricted with \(g_i=0\), while the remaining sets satisfy \(g_i\ge g_0>0\). This would permit analyses of the following form:
what if a specified fraction of matched sets are subject to the conventional analysis, while the rest obey a stochasticity restriction? Such a model could broaden the interpretability of the method in applications where the degree of adverse alignment is believed to vary across
matched sets.
It would also raise new questions about how to aggregate the least favorable distributions
across matched sets with different stochasticity levels.



\begingroup
%
\setlength{\bibsep}{0pt}
\bibliographystyle{apalike}
\bibliography{ref}
\endgroup

\newpage
\appendix

\begin{center}
\Huge{\textbf{Supplementary Material}}
\end{center}

This supplementary material contains the technical proofs and additional derivations supporting the main text. Specifically, Section~\ref{app:proof-justify} proves the exact finite-sample result for two-point statistics and the asymptotic upper-tail
validity result for general statistics. Section~\ref{app:proof-classes}
proves the characterization of least favorable distributions for the
mean-band, two-group, and Bernoulli classes. Section~\ref{app:more-classes}
gives additional results from Section~\ref{sec:classes}, including the proof of the claim in Example~\ref{Example-1}, a nonlinear
programming formulation for the full mean-band problem and a rescaled Beta subclass. Section~\ref{app:ds}
derives the design sensitivity formulas used in Section~\ref{sec:power} of
the main text.

\bigskip
\section{Proofs of results in Section~\ref{sec:justify-mean}}\label{app:proof-justify}
\subsection{Proof of Proposition~\ref{prop:sum-fsd}}

\begin{proof}[Proof of Proposition~\ref{prop:sum-fsd}]
Define
\[
H_i := \sum_{j=1}^{n_i} Z_{ij}\,\1\{q_{ij}=a_{i1}\}\in\{0,1\},
\]
so that
\[
T_i = a_{i2} + (a_{i1}-a_{i2})H_i \in \{a_{i2},a_{i1}\}.
\]
For any \(P_i\in{\cP}_i\), write
\[
p_i(P_i):=\P_{P_i}(T_i=a_{i1})=\P_{P_i}(H_i=1).
\]
If \(a_{i1}=a_{i2}\), then \(T_i\) is constant under every \(P_i\). We therefore
consider only sets with \(a_{i1}>a_{i2}\). For such sets,
\[
\mu_i(P_i)
=
a_{i2} + (a_{i1}-a_{i2})\,p_i(P_i).
\]
Therefore, because \(P_i^\ast\) maximizes \(\mu_i(P_i)\) over \({\cP}_i\), we have
\begin{equation}\label{eq:pi-order}
p_i(P_i^\ast)\ge p_i(P_i)
\qquad\text{for every }P_i\in{\cP}_i.
\end{equation}

Now fix an arbitrary product distribution \(P=\bigotimes_{i=1}^I P_i\in{\cP}\).
Let \(U_1,\dots,U_I\) be independent \(\mathrm{Unif}(0,1)\) variables. Define
\[
\widetilde T_i^\ast
:=
a_{i1}\,\1\{U_i\le p_i(P_i^\ast)\}
+
a_{i2}\,\1\{U_i>p_i(P_i^\ast)\},
\]
and
\[
\widetilde T_i
:=
a_{i1}\,\1\{U_i\le p_i(P_i)\}
+
a_{i2}\,\1\{U_i>p_i(P_i)\}.
\]
Then \(\widetilde T_i^\ast\) has the same distribution as \(T_i\) under \(P_i^\ast\), while \(\widetilde T_i\) has the same
distribution as \(T_i\) under \(P_i\). Moreover, by \eqref{eq:pi-order},
\[
\widetilde T_i^\ast \ge \widetilde T_i
\qquad\text{almost surely for each }i.
\]
Hence, defining
\[
\widetilde T^\ast:=\sum_{i=1}^I \widetilde T_i^\ast,
\qquad
\widetilde T:=\sum_{i=1}^I \widetilde T_i,
\]
we have
\[
\widetilde T^\ast\ge \widetilde T
\qquad\text{almost surely}.
\]
Since \(\widetilde T_i^\ast\) are independent and match the marginal
distributions of $T_i$ under \(P^\ast\), and the
\(\widetilde T_i\)'s are independent and match the marginal distributions
under \(P\), it follows that
\[
\widetilde T^\ast \stackrel{d}{=} T
\quad\text{under } P^\ast,
\qquad
\widetilde T \stackrel{d}{=} T
\quad\text{under } P.
\]
Therefore, for every \(a\in\mathbb R\),
\[
\P_{P^\ast}(T\ge a)
=
\P(\widetilde T^\ast\ge a)
\ge
\P(\widetilde T\ge a)
=
\P_P(T\ge a).
\]
Since \(P\in{\cP}\) was arbitrary, taking the supremum over \(P\in{\cP}\)
proves \eqref{eq:tail-prob-bound}.
\end{proof}
\subsection{Proof of Theorem~\ref{thm:asym-tail}}

\begin{proof}[Proof of Theorem~\ref{thm:asym-tail}]
Throughout the proof we condition on \((\mca,\cZ)\) and
suppress this conditioning from the notation.

Let
\[
m_I:=\sum_{i=1}^I \mu_i,
\qquad
m_I^\ast:=\sum_{i=1}^I \mu_i^\ast,
\qquad
\sigma_I^2:=\sum_{i=1}^I \nu_i^2,
\qquad
(\sigma_I^\ast)^2:=\sum_{i=1}^I(\nu_i^\ast)^2.
\]
The threshold is
\[
a_I=m_I+c\sigma_I.
\]

By Assumption~\ref{assump:asym-separable}(i),
\[
\sigma_I^2
=
\sum_{i=1}^I\nu_i^2
\ge I\underline\nu^2,
\qquad
(\sigma_I^\ast)^2
=
\sum_{i=1}^I(\nu_i^\ast)^2
\ge I\underline\nu^2.
\]
Also, by the moment bound in Assumption~\ref{assump:asym-separable}(i),
\[
\sum_{i=1}^I
\E_P\left[
|T_i-\mu_i|^{2+\zeta}
\right]
\le MI.
\]
Therefore the Lyapunov ratio under \(P\) satisfies
\[
\frac{
\sum_{i=1}^I
\E_P\left[
|T_i-\mu_i|^{2+\zeta}
\right]
}{
\sigma_I^{2+\zeta}
}
\le
\frac{MI}{(I\underline\nu^2)^{1+\zeta/2}}
=
\frac{M}{\underline\nu^{2+\zeta}}I^{-\zeta/2}
\longrightarrow 0.
\]
The same argument gives
\[
\frac{
\sum_{i=1}^I
\E_{P^\ast}\left[
|T_i-\mu_i^\ast|^{2+\zeta}
\right]
}{
(\sigma_I^\ast)^{2+\zeta}
}
\longrightarrow 0.
\]
Hence, by the Lyapunov central limit theorem,
\[
\frac{T-m_I}{\sigma_I}\Rightarrow N(0,1)
\quad\text{under }P,
\qquad
\frac{T-m_I^\ast}{\sigma_I^\ast}\Rightarrow N(0,1)
\quad\text{under }P^\ast.
\]

For every \(\epsilon>0\), there exists \(I_1\) such that, for all
\(I\ge I_1\),
\begin{equation}\label{eq:tail-under-P-corrected}
\P_P(T\ge a_I)
=
\P_P\left(
\frac{T-m_I}{\sigma_I}\ge c
\right)
\le
1-\Phi(c)+\epsilon/2,
\end{equation}
and
\begin{equation}\label{eq:tail-under-Pstar-corrected}
\P_{P^\ast}\left(T\ge m_I^\ast+c\sigma_I^\ast\right)
=
\P_{P^\ast}\left(
\frac{T-m_I^\ast}{\sigma_I^\ast}\ge c
\right)
\ge
1-\Phi(c)-\epsilon/2.
\end{equation}

We show that, for all sufficiently large \(I\),
\begin{equation}\label{eq:threshold-domination-corrected}
a_I=m_I+c\sigma_I
\le
m_I^\ast+c\sigma_I^\ast.
\end{equation}
Equivalently,
\begin{equation}\label{eq:threshold-domination-average-corrected}
\frac{1}{I}\sum_{i=1}^I(\mu_i^\ast-\mu_i)
\ge
\frac{c}{\sqrt I}
\left\{
\sqrt{\frac{1}{I}\sum_{i=1}^I\nu_i^2}
-
\sqrt{\frac{1}{I}\sum_{i=1}^I(\nu_i^\ast)^2}
\right\}.
\end{equation}

Let
\[
x_I:=\frac{1}{I}\sum_{i=1}^I\nu_i^2,
\qquad
y_I:=\frac{1}{I}\sum_{i=1}^I(\nu_i^\ast)^2.
\]
If \(x_I\le y_I\), then the right-hand side of
\eqref{eq:threshold-domination-average-corrected} is nonpositive, while
the left-hand side is nonnegative because \(\mu_i^\ast\ge \mu_i\) for
every \(i\). Thus the desired inequality holds.

It remains to consider \(x_I>y_I\). Since \(u\mapsto \sqrt u\) is
concave on \((0,\infty)\),
\[
\sqrt{x_I}-\sqrt{y_I}
\le
\frac{x_I-y_I}{2\sqrt{y_I}}.
\]
By Assumption~\ref{assump:asym-separable}(i), \(y_I\ge \underline\nu^2\),
so
\begin{equation}\label{eq:sqrt-diff-corrected}
\sqrt{x_I}-\sqrt{y_I}
\le
\frac{x_I-y_I}{2\underline\nu}.
\end{equation}

Now
\[
x_I-y_I
=
\frac{1}{I}\sum_{i=1}^I
\left\{\nu_i^2-(\nu_i^\ast)^2\right\}.
\]
For \(i\notin A_I(P)\), we have
\(\nu_i^2-(\nu_i^\ast)^2\le0\). Hence
\[
x_I-y_I
\le
\frac{1}{I}\sum_{i\in A_I(P)}
\left\{\nu_i^2-(\nu_i^\ast)^2\right\}.
\]
By Assumption~\ref{assump:asym-separable}(ii),
\[
\nu_i^2-(\nu_i^\ast)^2\le \bar\nu^2
\qquad
\text{for every }i\in A_I(P).
\]
Therefore
\begin{equation}\label{eq:variance-diff-corrected}
x_I-y_I
\le
\bar\nu^2\,\pi_I(P).
\end{equation}
Combining \eqref{eq:sqrt-diff-corrected} and
\eqref{eq:variance-diff-corrected},
\[
\sqrt{x_I}-\sqrt{y_I}
\le
\frac{\bar\nu^2}{2\underline\nu}\pi_I(P).
\]
Consequently, the right-hand side of
\eqref{eq:threshold-domination-average-corrected} is at most
\[
\frac{c\bar\nu^2}{2\underline\nu\sqrt I}\pi_I(P).
\]

By Assumption~\ref{assump:asym-separable}(ii),
\[
\frac{1}{I}\sum_{i=1}^I(\mu_i^\ast-\mu_i)
\ge
\delta\,\pi_I(P)
\]
for all sufficiently large \(I\). Choose \(I_2\) large enough that
\[
\frac{c\bar\nu^2}{2\underline\nu\sqrt I}
\le
\delta
\qquad
\text{for all } I\ge I_2.
\]
Then, for all sufficiently large \(I\),
\eqref{eq:threshold-domination-average-corrected} holds. Hence
\eqref{eq:threshold-domination-corrected} holds.

For all sufficiently large \(I\), by
\eqref{eq:threshold-domination-corrected},
\[
\P_{P^\ast}(T\ge a_I)
\ge
\P_{P^\ast}(T\ge m_I^\ast+c\sigma_I^\ast).
\]
Using \eqref{eq:tail-under-Pstar-corrected} and
\eqref{eq:tail-under-P-corrected}, we obtain
\[
\P_{P^\ast}(T\ge a_I)
\ge
1-\Phi(c)-\epsilon/2
\ge
\P_P(T\ge a_I)-\epsilon.
\]
This proves \eqref{eq:asym-tail-bound}.

Now suppose Fisher's sharp null \(H_F\) holds and \(P\) is the true distribution of the
hidden confounders. Fix \(\alpha\in(0,1/2)\), and define
\[
z_\alpha:=\Phi^{-1}(1-\alpha)>0.
\]
Therefore
\[
\{p^\ast(T_{\obs})\le \alpha\}
=
\left\{
T_{\obs}\ge m_I^\ast+z_\alpha\sigma_I^\ast
\right\}.
\]

Apply the threshold comparison above with \(c=z_\alpha\). Since
\(z_\alpha>0\), for all sufficiently large \(I\),
\[
m_I^\ast+z_\alpha\sigma_I^\ast
\ge
m_I+z_\alpha\sigma_I.
\]
Hence
\[
\P_P\{p^\ast(T_{\obs})\le\alpha\}
\le
\P_P\!\left\{
T_{\obs}\ge m_I+z_\alpha\sigma_I
\right\}.
\]
By the central limit theorem under \(P\),
\[
\P_P\!\left\{
T_{\obs}\ge m_I+z_\alpha\sigma_I
\right\}
=
\P_P\!\left\{
\frac{T_{\obs}-m_I}{\sigma_I}\ge z_\alpha
\right\}
\longrightarrow
1-\Phi(z_\alpha)
=
\alpha.
\]
Therefore
\[
\limsup_{I\to\infty}
\P_P\!\left\{
p^\ast(T_{\obs})\le\alpha
\mid \mca,\cZ
\right\}
\le \alpha.
\]
This proves the type-I error statement.
\end{proof}

\section{Proofs of results in Section~\ref{sec:classes}}\label{app:proof-classes}
\subsection{Proof of Theorem~\ref{thm:mix-product-mean}}
\begin{proof}[Proof of Theorem~\ref{thm:mix-product-mean}]
For any \(P_{\Lambda_i}\in\cP_i\),
\[
\mu_i(P_{\Lambda_i})
=
\int_{\mcl_i^{\otimes}}\mu_i(Q_i)\,d\Lambda_i(Q_i)
\le
\sup_{Q_i\in\mcl_i^{\otimes}}\mu_i(Q_i).
\]
Taking the supremum over \(P_{\Lambda_i}\in\cP_i\) gives
\[
\sup_{P_i\in\cP_i}\mu_i(P_i)
\le
\sup_{Q_i\in\mcl_i^{\otimes}}\mu_i(Q_i).
\]
The reverse inequality follows because every \(Q_i\in\mcl_i^{\otimes}\) can be
viewed as the degenerate mixture. Hence
\[
\sup_{P_i\in\cP_i}\mu_i(P_i)
=
\sup_{Q_i\in\mcl_i^{\otimes}}\mu_i(Q_i).
\]

Let
\[
M_i
:=
\sup_{Q_i\in\mcl_i^{\otimes}}\mu_i(Q_i).
\]
If \(\Lambda_i((\mcl_i^{\otimes})^\star)=1\), then
\[
\mu_i(P_{\Lambda_i})
=
\int_{\mcl_i^{\otimes}}\mu_i(Q_i)\,d\Lambda_i(Q_i)
=
M_i,
\]
so \(P_{\Lambda_i}\in\cP_i^\star\).

Conversely, suppose \(P_{\Lambda_i}\in\cP_i^\star\). Then
\[
0
=
M_i-\mu_i(P_{\Lambda_i})
=
\int_{\mcl_i^{\otimes}}
\{M_i-\mu_i(Q_i)\}\,d\Lambda_i(Q_i).
\]
The integrand is nonnegative everywhere on \(\mcl_i^{\otimes}\). Therefore it
must be zero \(\Lambda_i\)-almost surely, which implies
\[
\Lambda_i\bigl((\mcl_i^{\otimes})^\star\bigr)=1.
\]
This proves the characterization of \(\cP_i^\star\).

We now prove the variance assertion. Let
\(\widetilde Q_i\sim\Lambda_i\) denote the random product law in the
mixture representation. Conditional on \(\widetilde Q_i=Q_i\), the
set-specific statistic has mean \(\mu_i(Q_i)\) and variance
\(\nu_i^2(Q_i)\). Therefore, by the law of total variance,
\[
\nu_i^2(P_{\Lambda_i})
=
\int_{\mcl_i^{\otimes}}
\nu_i^2(Q_i)\,d\Lambda_i(Q_i)
+
\operatorname{Var}_{\Lambda_i}\{\mu_i(Q_i)\}.
\]
If \(P_{\Lambda_i}\in\cP_i^\star\), then \(\Lambda_i\) assigns probability one
to \((\mcl_i^{\otimes})^\star\). Hence
\[
\mu_i(Q_i)
=
\sup_{\widetilde Q_i\in\mcl_i^{\otimes}}\mu_i(\widetilde Q_i)
\qquad
\text{for } \Lambda_i\text{-almost every } Q_i.
\]
Thus \(\mu_i(Q_i)\) is constant \(\Lambda_i\)-almost surely, so
\[
\operatorname{Var}_{\Lambda_i}\{\mu_i(Q_i)\}=0.
\]
Moreover, since \(\Lambda_i\bigl((\mcl_i^{\otimes})^\star\bigr)=1\),
\[
\int_{\mcl_i^{\otimes}}
\nu_i^2(Q_i)\,d\Lambda_i(Q_i)
=
\int_{(\mcl_i^{\otimes})^\star}
\nu_i^2(Q_i)\,d\Lambda_i(Q_i).
\]
Therefore
\[
\nu_i^2(P_{\Lambda_i})
=
\int_{(\mcl_i^{\otimes})^\star}
\nu_i^2(Q_i)\,d\Lambda_i(Q_i).
\]

It follows that
\[
\nu_i^2(P_{\Lambda_i})
\le
\sup_{Q_i\in(\mcl_i^{\otimes})^\star}\nu_i^2(Q_i)
\qquad
\text{for every }P_{\Lambda_i}\in\cP_i^\star.
\]
Taking the supremum over \(P_{\Lambda_i}\in\cP_i^\star\) gives
\[
\sup_{P_i\in\cP_i^\star}\nu_i^2(P_i)
\le
\sup_{Q_i\in(\mcl_i^{\otimes})^\star}\nu_i^2(Q_i).
\]
The reverse inequality follows because every
\(Q_i\in(\mcl_i^{\otimes})^\star\) can be viewed as a degenerate mixture and
therefore belongs to \(\cP_i^\star\). Hence
\[
\sup_{P_i\in\cP_i^\star}\nu_i^2(P_i)
=
\sup_{Q_i\in(\mcl_i^{\otimes})^\star}\nu_i^2(Q_i).
\]
The final assertion follows immediately when the right-hand supremum is
attained at \(Q_i^\star\).
\end{proof}
\subsection{Proof of Proposition~\ref{prop:mean-band}}
\begin{lemma}\label{lm:extreme_L}
A distribution \(Q\in\mcl(g)\) is an extreme point of \(\mcl(g)\) if and only if
either

\smallskip

\noindent\textup{(i)} \(Q=\delta_c\) for some \(c\in[1,\Gamma]\cap[\mu^-(g),\mu^+(g)]\), or

\smallskip

\noindent\textup{(ii)} \(Q\) is supported on exactly two points in \([1,\Gamma]\) and
\[
\int x\,dQ(x)\in\{\mu^-(g),\mu^+(g)\}.
\]
\end{lemma}

\begin{proof}
Suppose first that \(Q\) has support containing at least three points. Then there exist
pairwise disjoint sets \(A_1,A_2,A_3\subseteq[1,\Gamma]\) such that
\(Q(A_r)>0\) for \(r=1,2,3\). Since the three vectors
\[
\left(Q(A_r),\int_{A_r}x\,dQ(x)\right)\in\mathbb R^2,
\qquad r=1,2,3,
\]
are linearly dependent, there exists
\((c_1,c_2,c_3)\neq(0,0,0)\) such that
\[
\sum_{r=1}^3 c_r\,Q(A_r)=0,
\qquad
\sum_{r=1}^3 c_r\int_{A_r}x\,dQ(x)=0.
\]
Define
\[
u(x):=\sum_{r=1}^3 c_r\,1_{A_r}(x).
\]
Then
\[
\int u(x)\,dQ(x)=0,
\qquad
\int x\,u(x)\,dQ(x)=0.
\]
For sufficiently small \(\varepsilon>0\), the functions
\(1\pm\varepsilon u(x)\) are nonnegative \(Q\)-a.s. Define
\[
Q^\pm(B):=\int_B \bigl(1\pm\varepsilon u(x)\bigr)\,dQ(x),
\qquad
B\subseteq[1,\Gamma].
\]
Since \((c_1,c_2,c_3)\neq 0\) and \(Q(A_r)>0\) for each \(r\), the
function \(u\) is not zero \(Q\)-almost surely. Hence, for sufficiently small
\(\varepsilon>0\), the measures \(Q^+\) and \(Q^-\) are distinct. Moreover, 
\[
\int x\,dQ^\pm(x)=\int x\,dQ(x).
\]
Hence \(Q^\pm\in\mcl(g)\) and
\[
Q=\frac12Q^+ + \frac12Q^-.
\]
Thus \(Q\) is not extreme. Therefore every extreme point of \(\mcl(g)\) has support
size at most two.

Next suppose
\[
Q=\lambda\delta_a+(1-\lambda)\delta_b,
\qquad 1\le a<b\le \Gamma,
\]
with \(\lambda\in(0,1)\), and
\[
\mu^-(g)<\lambda a+(1-\lambda)b<\mu^+(g).
\]
For sufficiently small \(\eta>0\), both \(\lambda\pm\eta\) lie in \([0,1]\), and both
means
\[
(\lambda\pm\eta)a+(1-\lambda\mp\eta)b
\]
still lie in \([\mu^-(g),\mu^+(g)]\). Thus the two distinct distributions
\[
Q^\pm
:=
(\lambda\pm\eta)\delta_a+(1-\lambda\mp\eta)\delta_b
\]
belong to \(\mcl(g)\), and
\[
Q=\frac12Q^+ + \frac12Q^-.
\]
So a two-point distribution with interior mean is not extreme.

Conversely, every point mass \(\delta_c\) with
\(c\in[1,\Gamma]\cap[\mu^-(g),\mu^+(g)]\) is extreme. Indeed, suppose
\[
\delta_c=tQ_1+(1-t)Q_2
\]
for some \(t\in(0,1)\) and \(Q_1,Q_2\in\mcl(g)\). Since \(\delta_c\) assigns no mass
to \([1,\Gamma]\setminus\{c\}\), the same is true of \(Q_1\) and \(Q_2\). Hence
\(Q_1=Q_2=\delta_c\).

Now let
\[
Q=\lambda\delta_a+(1-\lambda)\delta_b,
\qquad 1\le a<b\le \Gamma,
\]
with \(\lambda\in(0,1)\), and suppose
\[
\int x\,dQ(x)=\mu^-(g).
\]
The case \(\int x\,dQ(x)=\mu^+(g)\) is identical. Suppose
\[
Q=tQ_1+(1-t)Q_2
\]
for some \(t\in(0,1)\) and \(Q_1,Q_2\in\mcl(g)\). Since \(Q\) assigns no mass to
\([1,\Gamma]\setminus\{a,b\}\), the same is true of \(Q_1\) and \(Q_2\). Moreover,
\[
\mu^-(g)
=
\int x\,dQ(x)
=
t\int x\,dQ_1(x)+(1-t)\int x\,dQ_2(x),
\]
while feasibility implies
\[
\int x\,dQ_1(x)\ge \mu^-(g),
\qquad
\int x\,dQ_2(x)\ge \mu^-(g).
\]
Hence both means must equal \(\mu^-(g)\). But among probability distributions
supported on \(\{a,b\}\), the mean uniquely determines the mixing weight. Therefore
\(Q_1=Q_2=Q\), so \(Q\) is extreme.

This proves the lemma.
\end{proof}

\begin{proof}[Proof of Proposition~\ref{prop:mean-band}]
For \(\bx_i=(x_{i1},\ldots,x_{in_i})\in[1,\Gamma]^{n_i}\), define
\[
h_i(\bx_i)
:=
\frac{\sum_{j=1}^{n_i}q_{ij}x_{ij}}
     {\sum_{\ell=1}^{n_i}x_{i\ell}},
\]
so that $\mu_i(P_i)=\int h_i(\bx_i)\,dP_i(\bx_i).$

First, the supremum in \eqref{eq:obj_cpg} is attained. The set
\(\mcl(g)\) is weakly compact because it is a closed subset of
\(\mathcal P([1,\Gamma])\), and \([1,\Gamma]\) is compact. Hence
\(\prod_{j=1}^{n_i}\mcl(g)\) is compact. The map
\[
(Q_{i1},\ldots,Q_{in_i})
\mapsto
\int h_i(\bx_i)\,d\bigotimes_{j=1}^{n_i}Q_{ij}(\bx_i)
\]
is continuous because \(h_i\) is bounded and continuous on
\([1,\Gamma]^{n_i}\). Therefore a maximizer exists. Let $Q_i^\star
=
\bigotimes_{j=1}^{n_i}Q_{ij}^\star
\in \mcl_i^{\otimes}(g)$
be a maximizer of \eqref{eq:obj_cpg}.

We now show that the marginal laws \(Q_{ij}^\star\) may be chosen to
have the stated two-point form. Now fix \(j\in\{1,\ldots,n_i\}\), and fix the distributions
\[
Q_{i1}^\star,\ldots,Q_{i,j-1}^\star,Q_{i,j+1}^\star,\ldots,Q_{in_i}^\star.
\] 
For
\(x\in[1,\Gamma]\), define
\[
\phi_{ij}(x)
:=
\int
h_i(x_{i1},\ldots,x_{i,j-1},x,x_{i,j+1},\ldots,x_{in_i})
\prod_{\ell\ne j}Q_{i\ell}^\star(dx_{i\ell}).
\]
This is the conditional expectation of \(h_i(\bG_i)\) given \(G_{ij}=x\),
with the other coordinates integrated out under the fixed marginal laws.
Then, conditional on the other marginal laws, optimizing the \(j\)th marginal
amounts to the affine problem
\[
\sup_{Q\in\mcl(g)} \int \phi_{ij}(x)\,dQ(x).
\]
Because \(Q_i^\star\) is globally optimal, its \(j\)th marginal \(Q_{ij}^\star\) is a
maximizer of this one-dimensional problem. Since \(\mcl(g)\) is compact and convex and the objective is affine in \(Q\), by Bauer’s maximum principle, 
there exists an optimizer that is an extreme point of \(\mcl(g)\).

Replacing \(Q_{ij}^\star\) by such an extreme-point maximizer does
not decrease the objective.
Applying this argument successively for \(j=1,\ldots,n_i\), we obtain an optimizer,
still denoted by
\[
Q_i^\star=\bigotimes_{j=1}^{n_i}Q_{ij}^\star,
\]
such that every marginal \(Q_{ij}^\star\) is an extreme point of \(\mcl(g)\). This is a one-dimensional instance of the classical theory of extreme points of moment sets; see \cite{winkler88}.

By Lemma~\ref{lm:extreme_L}, each \(Q_{ij}^\star\) is supported on at most two points in \([1,\Gamma]\).
More precisely, either
\[
Q_{ij}^\star=\delta_c
\quad\text{for some }c\in[\mu^-(g),\mu^+(g)],
\]
or \(Q_{ij}^\star\) is supported on exactly two points and satisfies
\[
\E_{Q_{ij}^\star}[G_{ij}]
=
\int x\,dQ_{ij}^\star(x)
\in\{\mu^-(g),\mu^+(g)\}.
\]
Finally, by Theorem~\ref{thm:mix-product-mean}, the same product law
\(Q_i^\star\), viewed as the degenerate mixture, attains the worst-case mean over
\(\cP_i(g)\).
\end{proof}

\subsection{Proof of Theorem~\ref{thm:two-group}}
In this section we prove Theorem~\ref{thm:two-group}. 
We use Lemmas~\ref{lm:extreme_L}-\ref{lm:extreme-distributions-mean-band} to characterize the maximizers of 
\[
\sup_{Q_i\in\mcl_i^{\otimes,\mathrm{2G}}(g)}
\mu_i(Q_i).
\]
For a fixed product law \(Q_i = \bigotimes_{j=1}^{n_i}\, Q_{ij}\),
\[
\varrho_{ij}(Q_i):=\E_{Q_i}\!\bigl[\varrho_{ij}(\bG_i)\bigr],
\qquad
\mu_i(Q_i)=\sum_{j=1}^{n_i}\varrho_{ij}(Q_i)\,q_{ij}.
\]
Under the two-group restriction there exist two marginal distributions \(Q_i^+\) and \(Q_i^-\) such that each coordinate distribution \(Q_{ij}\) is equal to either \(Q_i^+\) or \(Q_i^-\). Let
\[
\mcm_i^+ := \{j:Q_{ij}=Q_i^+\},
\qquad
\mcm_i^- := [n_i]\setminus \mcm_i^+,
\qquad
k:=|\mcm_i^+|.
\]

The one-group cases \(k=0\) and \(k=n_i\) yield
\(\varrho_{ij}(Q_i)=1/n_i\) for all \(j\), and hence
\[
\mu_i(Q_i)=\frac{1}{n_i}\sum_{j=1}^{n_i}q_{ij}.
\]
These cases do not affect the supremum. Indeed, for any nontrivial
\(k\in\{1,\ldots,n_i-1\}\), the feasible set with a top-\(k\) partition
contains the one-group choice \(Q_i^+=Q_i^-\), which gives the same value
\((1/n_i)\sum_j q_{ij}\). Therefore the optimized value over nontrivial
two-group partitions is at least as large as the one-group value.

\begin{lemma}
\label{lm:two-valued-rho}
Fix \(Q_i^+\), \(Q_i^-\), and a partition \(\mcm_i^+\). Then \(\varrho_{ij}(Q_i)\) takes at most two values:
\[
\varrho_{ij}(Q_i)=\varrho_i^+(Q_i)\quad (j\in\mcm_i^+),
\qquad
\varrho_{ij}(Q_i)=\varrho_i^-(Q_i)\quad (j\in\mcm_i^-).
\]
Moreover, for fixed \((Q_i^+,Q_i^-)\), the pair \((\varrho_i^+(Q_i),\varrho_i^-(Q_i))\) depends on the partition only through \(|\mcm_i^+|\). 
\end{lemma}

\begin{proof}
Under the two-group restriction, the coordinates \(\{G_{ij}:j\in\mcm_i^+\}\) are i.i.d. with distribution \(Q_i^+\), the coordinates \(\{G_{ij}:j\in\mcm_i^-\}\) are i.i.d. with distribution \(Q_i^-\), and the two collections are independent. Since
\[
\varrho_{ij}(\bG_i)=\frac{G_{ij}}{\sum_{\ell=1}^{n_i}G_{i\ell}},
\]
swapping two indices within the same group leaves the joint distribution of \(\bG_i\) unchanged and swaps the corresponding treatment probabilities. Hence \(\varrho_{ij}(Q_i)\) is constant on \(\mcm_i^+\), and likewise constant on \(\mcm_i^-\), proving the first claim.

Now fix \((Q_i^+,Q_i^-)\) and two partitions of \([n_i]\) having the same cardinality. One partition is obtained from the other by a permutation of indices. Because the distribution of \(\bG_i\) is invariant under such a relabeling, the common values \(\varrho_i^+(Q_i)\) and \(\varrho_i^-(Q_i)\) depend on the partition only through the cardinality.
\end{proof}
Without loss of generality, we may assume \(\varrho_i^+(Q_i)\ge \varrho_i^-(Q_i)\) (relabeling the two groups if necessary).
\begin{lemma}
\label{lm:Stoch_topk}
Fix \(Q_i^+\), \(Q_i^-\), and consider all partitions \(\mcm_i^+\subset[n_i]\) of size \(k=|\mcm_i^+|\). If \(q_{i1}\ge \cdots \ge q_{in_i}\), then, among all partitions of size \(k\), \(\mu_i(Q_i)\) is maximized by
\(\mcm_i^+=\{1,\ldots,k\}\).
\end{lemma}

\begin{proof}
By Lemma~\ref{lm:two-valued-rho},
\[
\mu_i(Q_i)
=
\sum_{j\in\mcm_i^+}\varrho_i^+(Q_i)q_{ij}
+
\sum_{j\in\mcm_i^-}\varrho_i^-(Q_i)q_{ij}
=
\varrho_i^-(Q_i)\sum_{j=1}^{n_i}q_{ij}
+
\bigl(\varrho_i^+(Q_i)-\varrho_i^-(Q_i)\bigr)\sum_{j\in\mcm_i^+}q_{ij}.
\]
The first term does not depend on the partition, and the coefficient of the second term is nonnegative. Hence, among all subsets of size \(k\), the value of \(\mu_i(Q_i)\) is maximized by choosing the \(k\) largest \(q_{ij}\)'s, namely by taking \(\mcm_i^+=\{1,\ldots,k\}\).
\end{proof}

\begin{lemma}
\label{lm:max_rho_plus}
Fix \(k\in\{1,\ldots,n_i-1\}\) and the top-\(k\) partition \(\mcm_i^+(k):=\{1,\ldots,k\}\). Let
\[
\Psi_i(k;Q_i^+,Q_i^-)
:=
\E_{Q_i}\!\left[
  \frac{\sum_{j=1}^k G_{ij}}{\sum_{j=1}^{n_i} G_{ij}}
\right]
=
 k\,\varrho_i^+(Q_i).
\]
Then any maximizer of \(\Psi_i(k;Q_i^+,Q_i^-)\) over admissible \((Q_i^+,Q_i^-)\) is also a maximizer of \(\mu_i(Q_i)\) among two-group distributions with \(|\mcm_i^+|=k\).
\end{lemma}

\begin{proof}
Since \(\sum_{j=1}^{n_i}\varrho_{ij}(\bG_i)=1\) almost surely, Lemma~\ref{lm:two-valued-rho} gives
\[
 k\,\varrho_i^+(Q_i)+(n_i-k)\,\varrho_i^-(Q_i)=1.
\]
Hence
\[
\mu_i(Q_i)
=
\frac{1}{n_i-k}\sum_{j=k+1}^{n_i} q_{ij}
+
 k\,\varrho_i^+(Q_i)
\left(
\frac{1}{k}\sum_{j=1}^{k} q_{ij}
-
\frac{1}{n_i-k}\sum_{j=k+1}^{n_i} q_{ij}
\right).
\]
Because \(q_{i1}\ge\cdots\ge q_{in_i}\), the coefficient multiplying \(k\varrho_i^+(Q_i)\) is nonnegative. Therefore any choice of \((Q_i^+,Q_i^-)\) that maximizes \(\Psi_i(k;Q_i^+,Q_i^-)\) also maximizes \(\mu_i(Q_i)\) for this fixed value of \(k\).
\end{proof}

\begin{lemma}
\label{lm:extreme-distributions-mean-band}
Fix \(k\in\{1,\ldots,n_i-1\}\). Over all admissible pairs \((Q_i^+,Q_i^-)\), the functional \(\Psi_i(k;Q_i^+,Q_i^-)\) is maximized at
\[
Q_i^{+\star}=\delta_{\mu^+(g)},
\qquad
Q_i^{-\star}=(1-g)\delta_1+g\delta_\Gamma.
\]
\end{lemma}

\begin{proof}
Write
\[
A_k:=\sum_{j=1}^k G_{ij},
\qquad
B_k:=\sum_{j=k+1}^{n_i} G_{ij},
\qquad
\Psi_i(k;Q_i^+,Q_i^-)=\E_{Q_i}\!\left[\frac{A_k}{A_k+B_k}\right].
\]
We optimize the plus and minus distributions separately.
In the following argument, we temporarily allow the marginal distributions within a group to differ. This only enlarges the feasible set, so any upper bound obtained in this enlarged product class applies to the two-group subclass. The final optimal distribution has common marginals within each group and therefore belongs to \(\mcl_i^{\otimes,\mathrm{2G}}(g)\).

First consider a plus-group coordinate \(j\le k\). Conditional on all coordinates except \(G_{ij}\), write
\[
C_j:=\sum_{\ell\le k,\,\ell\ne j}G_{i\ell},
\qquad
D_j:=\sum_{\ell>k}G_{i\ell}.
\]
Then
\[
x\longmapsto \frac{x+C_j}{x+C_j+D_j}
\]
is increasing and concave on \([1,\Gamma]\). Let \(m_+:=\E_{Q_i^+}[G_{ij}]\in[\mu^-(g),\mu^+(g)]\). By Jensen's inequality,
\[
\E_{Q_i}\!\left[\left.\frac{G_{ij}+C_j}{G_{ij}+C_j+D_j}\,\right|\,C_j,D_j\right]
\le
\frac{m_+ + C_j}{m_+ + C_j + D_j}
\le
\frac{\mu^+(g) + C_j}{\mu^+(g) + C_j + D_j}.
\]
Thus replacing \(G_{ij}\) by the constant \(\mu^+(g)\) weakly increases \(\Psi_i(k;Q_i^+,Q_i^-)\). Repeating this argument for every \(j\le k\) yields
\[
\Psi_i(k;Q_i^+,Q_i^-)
\le
\Psi_i\bigl(k;\delta_{\mu^+(g)},Q_i^-\bigr).
\]

Now consider a minus-group coordinate \(j>k\). Conditional on all coordinates except \(G_{ij}\), write
\[
A_j:=\sum_{\ell\le k}G_{i\ell},
\qquad
B_j:=\sum_{\ell>k,\,\ell\ne j}G_{i\ell}.
\]
Then
\[
y\longmapsto \frac{A_j}{A_j+B_j+y}
\]
is decreasing and convex on \([1,\Gamma]\). Let \(m_-:=\E_{Q_i^-}[G_{ij}]\in[\mu^-(g),\mu^+(g)]\). For any \(Y\sim Q_i^-\), by the Edmundson--Madansky inequality \citep{madansky59},
\[
\E\!\left[\left.\frac{A_j}{A_j+B_j+Y}\,\right|\,A_j,B_j\right]
\le
\frac{\Gamma-m_-}{\Gamma-1}\cdot \frac{A_j}{A_j+B_j+1}
+
\frac{m_- -1}{\Gamma-1}\cdot \frac{A_j}{A_j+B_j+\Gamma}.
\]
The right-hand side is attained by the endpoint Bernoulli distribution on \(\{1,\Gamma\}\) with mean \(m_-\). Hence replacing each minus-group coordinate by that endpoint Bernoulli distribution weakly increases \(\Psi_i\). Moreover, the right-hand side is affine and decreasing in \(m_-\), so it is maximized at the smallest admissible mean, namely \(m_-=\mu^-(g)\). Since $\mu^-(g)=1+(\Gamma-1)g,$
the corresponding endpoint Bernoulli distribution is exactly
\[
(1-g)\delta_1+g\delta_\Gamma.
\]
Therefore
\[
\Psi_i\bigl(k;\delta_{\mu^+(g)},Q_i^-\bigr)
\le
\Psi_i\bigl(k;\delta_{\mu^+(g)},(1-g)\delta_1+g\delta_\Gamma\bigr),
\]
which proves the claim.
\end{proof}

\begin{proof}[Proof of Theorem~\ref{thm:two-group}]
Fix \(k\in\{1,\ldots,n_i-1\}\). By Lemma~\ref{lm:Stoch_topk}, among all
two-group product distributions with \(|\mcm_i^+|=k\), it suffices to consider the
top-\(k\) partition \(\mcm_i^+(k)=\{1,\ldots,k\}\). By
Lemma~\ref{lm:extreme-distributions-mean-band}, within this class there is a
maximizer of \(\Psi_i(k;Q_i^+,Q_i^-)\), and hence by
Lemma~\ref{lm:max_rho_plus} also of \(\mu_i(Q_i)\), given by $Q_i^{(k)}$.
Since \(k\) ranges over the finite set \(\{1,\ldots,n_i-1\}\), it follows that
\[
\sup_{Q_i\in\mcl_i^{\otimes,\mathrm{2G}}(g)}
\mu_i(Q_i)
=
\max_{1\le k\le n_i-1}
\mu_i\!\left(Q_i^{(k)}\right).
\]
By Theorem~\ref{thm:mix-product-mean}, the same value is the worst-case
mean over \(\cP_i^{\mathrm{2G}}(g)\), and a maximizing product law can be
viewed as a degenerate mixture.

It remains to identify the variance tie-breaker. The preceding argument shows
that, up to permutations of indices with tied \(q_{ij}\)'s, every
mean-maximizing product law in
\(\mcl_i^{\otimes,\mathrm{2G}}(g)\) is one of the finitely many laws
\(\{Q_i^{(k)}:k\in\mck_i\}\). Therefore, by the definition of
\(k_i^\star\),
\[
\nu_i^2(Q_i^{(k_i^\star)})
=
\max_{Q_i\in(\mcl_i^{\otimes,\mathrm{2G}}(g))^\star}
\nu_i^2(Q_i).
\]
The variance assertion over the mixture class then follows from
Theorem~\ref{thm:mix-product-mean}.
\end{proof}

\subsection{Proof of Theorem~\ref{thm:ber}}

\begin{lemma}\label{lem:multiaffine-endpoints}
Let \(F\) and \(W\) be real-valued functions on \([g,1-g]^d\) that are
affine in each coordinate separately. Then the maximum of \(F\) over
\([g,1-g]^d\) is attained at an endpoint vector, that is, at a vector
\(\mathbf p\) satisfying
\[
p_j\in\{g,1-g\},\qquad j=1,\ldots,n.
\]
Moreover, among the maximizers of \(F\), a maximizer of \(W\) can also be
chosen to be an endpoint vector.
\end{lemma}

\begin{proof}
The first claim follows by optimizing one coordinate at a time. Fix all
coordinates except \(p_j\). Since \(F\) is affine in \(p_j\), its maximum over
\([g,1-g]\) is attained at one of the endpoints \(g\) or \(1-g\). Repeating
this argument for \(j=1,\ldots,d\) gives an endpoint vector that maximizes
\(F\).

For the second claim, let \(\mathbf p\) be a maximizer of \(F\) that also
maximizes \(W\) among all maximizers of \(F\). If all coordinates of
\(\mathbf p\) are endpoints, there is nothing to prove. Otherwise, because
\(F\) is affine in each coordinate separately, \(F(\mathbf p)\) can be written
as a convex combination of the values of \(F\) at the endpoint vectors obtained
by replacing each non-endpoint coordinate of \(\mathbf p\) by either \(g\) or
\(1-g\). Since \(\mathbf p\) maximizes \(F\), every endpoint vector appearing
with positive weight in this convex combination must also maximize \(F\).

The same convex-combination representation holds for \(W(\mathbf p)\), since
\(W\) is also affine in each coordinate separately. Therefore \(W(\mathbf p)\)
is a convex combination of the values of \(W\) at endpoint vectors that also
maximize \(F\). At least one of these endpoint vectors has \(W\)-value no
smaller than \(W(\mathbf p)\). Since \(\mathbf p\) was chosen to maximize
\(W\) among the maximizers of \(F\), that endpoint vector also maximizes
\(W\) among the maximizers of \(F\).
\end{proof}

\begin{proof}[Proof of Theorem~\ref{thm:ber}]
Recall that $q_{i1}\ge q_{i2}\ge \cdots \ge q_{in_i}$ and \(q_{i1}>q_{in_i}\).
Under the Bernoulli subclass, we may write
\[
G_{ij}=1+(\Gamma-1)B_{ij},
\qquad
B_{ij}\stackrel{\mathrm{ind}}{\sim}\mathrm{Bernoulli}(p_{ij}),
\qquad
p_{ij}\in[g,1-g],
\]
for \(j=1,\dots,n_i\). Hence
\[
\mu_i(Q_i)
=
\E\!\left[
\frac{\sum_{j=1}^{n_i}\{1+(\Gamma-1)B_{ij}\}q_{ij}}
     {\sum_{j=1}^{n_i}\{1+(\Gamma-1)B_{ij}\}}
\right].
\]

For \(\mathbf b_i=(b_{i1},\dots,b_{in_i})\in\{0,1\}^{n_i}\), define
\[
H_i(\mathbf b_i)
:=
\frac{\sum_{j=1}^{n_i}\{1+(\Gamma-1)b_{ij}\}q_{ij}}
     {\sum_{j=1}^{n_i}\{1+(\Gamma-1)b_{ij}\}}.
\]
Then
\[
\mu_i(Q_i)
=
\sum_{\mathbf b_i\in\{0,1\}^{n_i}}
H_i(\mathbf b_i)\prod_{j=1}^{n_i}
p_{ij}^{\,b_{ij}}(1-p_{ij})^{1-b_{ij}}.
\]
Thus \(\mu_i(Q_i)\) is continuous on the compact set \([g,1-g]^{n_i}\) and
is affine in each coordinate \(p_{ij}\) separately. By
Lemma~\ref{lem:multiaffine-endpoints}, a mean maximizer can be chosen with
\[
p_{ij}^\star\in\{g,1-g\},\qquad j=1,\ldots,n_i.
\]

It remains to show that the larger value \(1-g\) can be assigned to the larger scores. Let
\[
\mathbf p_i=(p_{i1},\dots,p_{in_i})\in\{g,1-g\}^{n_i},
\]
and suppose that for some \(a<b\),
\[
p_{ia}=g,\qquad p_{ib}=1-g.
\]
Let \(\mathbf p_i'\) be obtained by swapping these two coordinates:
\[
p_{ia}'=1-g,\qquad p_{ib}'=g,\qquad p_{i\ell}'=p_{i\ell}\ \text{for }\ell\neq a,b.
\]
We claim that \(\mu_i(Q_i')\ge \mu_i(Q_i)\).

Condition on \(\{B_{i\ell}:\ell\neq a,b\}\), and define
\[
S_{i,-ab}:=\sum_{\ell\neq a,b} G_{i\ell},
\qquad
A_{i,-ab}:=\sum_{\ell\neq a,b} G_{i\ell}q_{i\ell},
\]
together with
\[
h_{ab}(x,y)
:=
\frac{xq_{ia}+yq_{ib}+A_{i,-ab}}{x+y+S_{i,-ab}},
\qquad x,y\in\{1,\Gamma\}.
\]
Given \(\{B_{i\ell}:\ell\neq a,b\}\), the only configurations affected by the swap are
\((G_{ia},G_{ib})=(\Gamma,1)\) and \((1,\Gamma)\). Therefore,
\[
\E_{\mathbf p_i'}\!\bigl[H_i(\mathbf B_i)\mid \{B_{i\ell}:\ell\neq a,b\}\bigr]
-
\E_{\mathbf p_i}\!\bigl[H_i(\mathbf B_i)\mid \{B_{i\ell}:\ell\neq a,b\}\bigr]
=
(1-2g)\bigl\{h_{ab}(\Gamma,1)-h_{ab}(1,\Gamma)\bigr\}.
\]
Since
\[
h_{ab}(\Gamma,1)-h_{ab}(1,\Gamma)
=
\frac{(\Gamma-1)(q_{ia}-q_{ib})}{\Gamma+1+S_{i,-ab}}
\ge 0,
\]
and \(1-2g\ge 0\), it follows that
\[
\E_{\mathbf p_i'}\!\bigl[H_i(\mathbf B_i)\mid \{B_{i\ell}:\ell\neq a,b\}\bigr]
\ge
\E_{\mathbf p_i}\!\bigl[H_i(\mathbf B_i)\mid \{B_{i\ell}:\ell\neq a,b\}\bigr].
\]
Taking expectations over \(\{B_{i\ell}:\ell\neq a,b\}\) gives
\[
\mu_i(Q_i')\ge \mu_i(Q_i).
\]

Thus, whenever a larger score is paired with \(g\) and a smaller score is paired with \(1-g\), swapping the two cannot decrease the objective. Repeating this pairwise swap argument yields an optimizer of top-\(k\) form:
\[
p_{ij}^\star=1-g \quad \text{for } j\le k,
\qquad
p_{ij}^\star=g \quad \text{for } j>k,
\]
for some \(k\in\{0,\dots,n_i\}\).

It remains to rule out the edge cases \(k=0\) and \(k=n_i\). First suppose \(g<1/2\). If \(k=0\), then \(p_{ij}^\star=g\) for all \(j\). Conditioning on \(\{B_{i2},\dots,B_{in_i}\}\), define
\[
S_{i,-1}:=\sum_{\ell=2}^{n_i} G_{i\ell},
\qquad
A_{i,-1}:=\sum_{\ell=2}^{n_i} G_{i\ell}q_{i\ell},
\]
and
\[
h_1(x):=\frac{xq_{i1}+A_{i,-1}}{x+S_{i,-1}},
\qquad x\in\{1,\Gamma\}.
\]
Then
\[
h_1(\Gamma)-h_1(1)
=
\frac{(\Gamma-1)\sum_{\ell=2}^{n_i}G_{i\ell}(q_{i1}-q_{i\ell})}
     {(\Gamma+S_{i,-1})(1+S_{i,-1})}
>0,
\]
because \(q_{i1}>q_{in_i}\) and all \(G_{i\ell}>0\). Hence increasing \(p_{i1}\) from \(g\) to \(1-g\) strictly increases \(\mu_i(Q_i)\), contradicting optimality. So \(k\neq 0\).

Similarly, if \(k=n_i\), then \(p_{ij}^\star=1-g\) for all \(j\). Conditioning on \(\{B_{i1},\dots,B_{i,n_i-1}\}\), define
\[
S_{i,-n_i}:=\sum_{\ell=1}^{n_i-1} G_{i\ell},
\qquad
A_{i,-n_i}:=\sum_{\ell=1}^{n_i-1} G_{i\ell}q_{i\ell},
\]
and
\[
h_{n_i}(x):=\frac{xq_{i,n_i}+A_{i,-n_i}}{x+S_{i,-n_i}},
\qquad x\in\{1,\Gamma\}.
\]
Then
\[
h_{n_i}(\Gamma)-h_{n_i}(1)
=
\frac{(\Gamma-1)\sum_{\ell=1}^{n_i-1}G_{i\ell}(q_{i,n_i}-q_{i\ell})}
     {(\Gamma+S_{i,-n_i})(1+S_{i,-n_i})}
<0,
\]
so decreasing \(p_{i,n_i}\) from \(1-g\) to \(g\) strictly increases \(\mu_i(Q_i)\), again contradicting optimality. Therefore \(k\neq n_i\).

If \(g=1/2\), then \(g=1-g\), so every coordinate equals \(1/2\) and the displayed top-\(k\) form holds trivially for any \(k\in\{1,\dots,n_i-1\}\).

Thus, in all cases, an optimizer can be chosen so that, for some \(k\in\{1,\dots,n_i-1\}\),
\[
p_{ij}^\star=1-g \quad \text{for } j\le k,
\qquad
p_{ij}^\star=g \quad \text{for } j>k.
\]
Therefore
\[
\sup_{Q_i\in\mcl_i^{\otimes,\mathrm{Bern}}(g)}
\mu_i(Q_i)
=
\max_{1\le k\le n_i-1}\mu_i(Q_i^{(k)}),
\]
where \(Q_i^{(k)}\) has marginals
\(\operatorname{Bern}_{1,\Gamma}(1-g)\) for \(j\le k\) and
\(\operatorname{Bern}_{1,\Gamma}(g)\) for \(j>k\). By Theorem~\ref{thm:mix-product-mean}, the same value is the worst-case
mean over \(\cP_i^{\mathrm{Bern}}(g)\), and a maximizing product law can be
viewed as a degenerate mixture.

It remains to identify the variance tie-breaker. The preceding argument shows
that the maximum of \(\mu_i(Q_i)\) is attained by at least one top-\(k\) law: \(Q_i^{(k)}\). We now show that the variance tie-breaker can also be
resolved among these laws.

For a Bernoulli product law \(Q_i\), let
\[
W_i(Q_i):=
\sum_{j=1}^{n_i}\varrho_{ij}(Q_i)q_{ij}^2.
\]
The same finite-sum representation, with \(q_{ij}^2\) in place of \(q_{ij}\),
shows that \(W_i(Q_i)\) is also affine in each coordinate \(p_{ij}\)
separately. Since
\[
\nu_i^2(Q_i)=W_i(Q_i)-\mu_i(Q_i)^2,
\]
and \(\mu_i(Q_i)\) is constant over the mean-maximizing class, maximizing
\(\nu_i^2(Q_i)\) among mean maximizers is equivalent to maximizing
\(W_i(Q_i)\) among mean maximizers. By Lemma~\ref{lem:multiaffine-endpoints},
a variance-maximizing mean maximizer can be chosen with
\[
p_{ij}\in\{g,1-g\},\qquad j=1,\ldots,n_i.
\]
The swap argument above then shows that any endpoint law that maximizes the
mean can be rearranged into top-\(k\) form without decreasing the mean. Since a
variance-maximizing mean maximizer is, in particular, a mean maximizer, it has
no inversion across a strict inequality \(q_{ia}>q_{ib}\). Any remaining
inversion can only occur among tied scores, and swapping tied scores changes
neither \(\mu_i(Q_i)\) nor \(W_i(Q_i)\), hence neither \(\nu_i^2(Q_i)\).
Therefore a variance-maximizing mean maximizer can be chosen in top-\(k\)
form. Therefore a variance-maximizing mean maximizer can be chosen from
\(\{Q_i^{(k)}:k\in\mck_i\}\). By the definition of \(k_i^\star\),
\[
\nu_i^2(Q_i^{(k_i^\star)})
=
\max_{Q_i\in(\mcl_i^{\otimes,\mathrm{Bern}}(g))^\star}
\nu_i^2(Q_i).
\]
Theorem~\ref{thm:mix-product-mean} then implies that \(Q_i^{(k_i^\star)}\), viewed as a degenerate mixture, is least favorable.
\end{proof}

\section{More results from Section~\ref{sec:classes}}\label{app:more-classes}
\subsection{Product law of \(\bG_i\) given a fixed treatment assignment}\label{app:mixture-generative}
We justify the claim that $\P(\bG_i\in\cdot\mid \mca,\cZ,\bZ_i=\bz_i)$ is a product law in Section~\ref{sec:classes}. Fix a
matched set \(i\), and suppose that the coordinates of \(\bG_i\) are
independent conditional on \(\mca\). Thus, for a generic realization
\(\bxi_i=(\xi_{i1},\ldots,\xi_{in_i})\),
\[
\P(\bG_i\in d\bxi_i\mid \mca)
=
\bigotimes_{j=1}^{n_i} P_{ij}(d\xi_{ij})
\]
for some one-dimensional laws \(P_{ij}\) on \([1,\Gamma]\).

Let $\lambda_{ij}:=\exp\{\kappa(\bx_{ij})\}.$
Before conditioning on the matched design, assume, as in the  sensitivity model \citep{ros87}, that treatment assignments are
conditionally independent given \((\mca,\bG_i)\), with
\[
\P(Z_{ij}=1\mid \mca,\bG_i)
=
\frac{\lambda_{ij}G_{ij}}{1+\lambda_{ij}G_{ij}}.
\]
For \(z\in\{0,1\}\), define
\[
h_{ij}^{z}(\xi)
:=
\left(
\frac{\lambda_{ij}\xi}{1+\lambda_{ij}\xi}
\right)^z
\left(
\frac{1}{1+\lambda_{ij}\xi}
\right)^{1-z}.
\]
Then, for any fixed assignment vector
\(\bz_i=(z_{i1},\ldots,z_{in_i})\),
Bayes' rule gives
\[
\P(\bG_i\in d\bxi_i\mid \mca,\bZ_i=\bz_i)
\propto
\prod_{j=1}^{n_i}
h_{ij}^{z_{ij}}(\xi_{ij})\,P_{ij}(d\xi_{ij}).
\]
The right-hand side factorizes across \(j\). Equivalently, if
\[
m_{ij}^{z}
:=
\int h_{ij}^{z}(\xi)\,P_{ij}(d\xi),
\qquad
P_{ij}^{z}(d\xi)
:=
\frac{h_{ij}^{z}(\xi)P_{ij}(d\xi)}{m_{ij}^{z}},
\]
then
\[
\P(\bG_i\in d\bxi_i\mid \mca,\bZ_i=\bz_i)
=
\bigotimes_{j=1}^{n_i}
P_{ij}^{z_{ij}}(d\xi_{ij}).
\]
Under the assumed independence across matched sets, conditioning additionally
on the full matched-design event \(\cZ\) does not change this set-specific
conditional law once \(\bZ_i=\bz_i\) is fixed. Hence, for each
\(\bz_i\in\Omega_i\),
\[
\P(\bG_i\in d\bxi_i\mid \mca,\cZ,\bZ_i=\bz_i)
=
\bigotimes_{j=1}^{n_i}
P_{ij}^{z_{ij}}(d\xi_{ij}),
\]
so the conditional law is a product law.

\subsection{Numerical solution of the optimal distributions in
Proposition~\ref{prop:mean-band}}
\label{app:nonlinear_pro}

This section gives a nonlinear programming formulation for
computing an optimal product law in Proposition~\ref{prop:mean-band}. For
notational convenience, we suppress the dependence on \(i\) and write
\(J:=n_i\). By Proposition~\ref{prop:mean-band}, it is enough to consider
marginal distributions supported on at most two points. We therefore
parameterize
\[
G_j=
\begin{cases}
a_j, & \text{with probability } 1-p_j,\\
b_j, & \text{with probability } p_j,
\end{cases}
\qquad
1 \le a_j \le b_j \le \Gamma,\quad 0 \le p_j \le 1,
\quad j=1,\dots,J.
\]
The degenerate case is included by allowing \(a_j=b_j\). The mean-band
constraints become
\[
\mu^{-}(g)\le (1-p_j)a_j + p_j b_j \le \mu^{+}(g),
\qquad j=1,\dots,J.
\]

For \(z=(z_1,\dots,z_J)\in\{0,1\}^J\), define
\[
s_j(z_j):=a_j + (b_j-a_j)z_j.
\]
Since the coordinates are independent,
\[
\E\!\left[\sum_{j=1}^{J}
\frac{G_j}{\sum_{\ell=1}^{J}G_\ell}\,q_j\right]
=
\sum_{z\in\{0,1\}^J}
\left(\prod_{j=1}^J p_j^{z_j}(1-p_j)^{1-z_j}\right)\,
\frac{\sum_{j=1}^J q_j\,s_j(z_j)}
{\sum_{j=1}^J s_j(z_j)}.
\]
Hence, by Proposition~\ref{prop:mean-band}, the original infinite-dimensional problem has the same optimal value as the following smooth nonconvex nonlinear program:
\[
\begin{aligned}
\max_{a,b,p}\quad
&\sum_{z\in\{0,1\}^J}
\left(\prod_{j=1}^J p_j^{z_j}(1-p_j)^{1-z_j}\right)
\frac{\sum_{j=1}^J q_j\bigl(a_j+(b_j-a_j)z_j\bigr)}
{\sum_{j=1}^J \bigl(a_j+(b_j-a_j)z_j\bigr)}\\
\text{s.t.}\quad
&1\le a_j\le b_j\le \Gamma,\qquad j=1,\dots,J,\\
&0\le p_j\le 1,\qquad j=1,\dots,J,\\
&\mu^{-}(g)\le (1-p_j)a_j+p_j b_j\le \mu^{+}(g),
\qquad j=1,\dots,J.
\end{aligned}
\tag{NLP}
\]

Because \(G_j\in[1,\Gamma]\), every denominator in \((\mathrm{NLP})\) is bounded
below by \(J\), so the objective is smooth on a compact feasible set.

For implementation it is cleaner to reparameterize
\[
b_j=a_j+d_j,\qquad 0\le d_j\le \Gamma-a_j,
\]
so that the mean constraint becomes
\[
\mu^{-}(g)\le a_j+p_j d_j\le \mu^{+}(g).
\]
This removes the ordering constraint \(a_j\le b_j\).

When \(J\) is small or moderate, the objective may be evaluated exactly by summing over all \(2^J\) support configurations of \((G_1,\dots,G_J)\). 
\subsection{A rescaled Beta subclass}
\label{subsec:stoch-beta}
Here we consider a continuous parametric subclass in which each \(G_{ij}\)
follows a Beta distribution affinely rescaled to \([1,\Gamma]\). Let
\(\mathrm{Beta}_{[1,\Gamma]}(\alpha,\beta)\) denote the distribution of
\(1+(\Gamma-1)B\), where \(B\sim \mathrm{Beta}(\alpha,\beta)\). If
\(G\sim \mathrm{Beta}_{[1,\Gamma]}(\alpha,\beta)\), then
\[
\E[G]
=
1+(\Gamma-1)\frac{\alpha}{\alpha+\beta}.
\]
Define \(\mcl_i^{\otimes,\mathrm{Beta}}(g)\) as the subclass of
\(\mcl_i^{\otimes,\mathrm{2G}}(g)\) consisting of product laws
\(Q_i=\bigotimes_{j=1}^{n_i}Q_{ij}\) for which there exist parameter pairs
\[
(\alpha^+,\beta^+)\in(0,\infty)^2,
\qquad
(\alpha^-,\beta^-)\in(0,\infty)^2,
\]
such that
\[
Q_{ij}\in
\left\{
\mathrm{Beta}_{[1,\Gamma]}(\alpha^+,\beta^+),
\mathrm{Beta}_{[1,\Gamma]}(\alpha^-,\beta^-)
\right\},
\qquad
j=1,\ldots,n_i,
\]
and
\begin{equation}\label{eq:cons-beta}
g \le \frac{\alpha^+}{\alpha^+ + \beta^+} \le 1-g,
\qquad
g \le \frac{\alpha^-}{\alpha^- + \beta^-} \le 1-g.
\end{equation}
Thus, \(\mcl_i^{\otimes,\mathrm{Beta}}(g)\) is a two-group parametric subclass
whose two possible marginal distributions are rescaled Beta distributions on
\([1,\Gamma]\).

Our next result shows that the rescaled Beta family can approximate the two-group analysis arbitrarily well, in the sense that the worst-case means are the same.
\begin{proposition}
    \label{prop:beta-opt}
Suppose that \(q_{i1}\ge \cdots \ge q_{in_i}\). Then
\begin{equation}\label{eq:sup-beta-equals-sup}
\sup_{Q_i\in \mcl_i^{\otimes,\mathrm{Beta}}(g)} \mu_i(Q_i)
\;=\;
\sup_{Q_i\in \mcl_i^{\otimes,\mathrm{2G}}(g)} \mu_i(Q_i).
\end{equation}
\end{proposition}

The point is simple. By Theorem~\ref{thm:two-group}, the optimizer in the two-group class uses the two boundary distributions
\[
\delta_{\mu^+(g)}
\qquad\text{and}\qquad
(1-g)\delta_1+g\delta_\Gamma.
\]
The rescaled Beta family can approximate both of them arbitrarily well: the first by sending the concentration parameter to infinity while keeping the mean fixed, and the second by sending both shape parameters to zero in the ratio \(g:(1-g)\). Hence the Beta subclass has the same supremal value, even though the supremum need not be attained by an interior Beta distribution. We now prove the proposition. 

\begin{proof}
Because \(\mcl_i^{\otimes,\mathrm{Beta}}(g)\subseteq \mcl_i^{\otimes,\mathrm{2G}}(g)\), we immediately have
\[
\sup_{Q_i\in \mcl_i^{\otimes,\mathrm{Beta}}(g)} \mu_i(Q_i)
\;\le\;
\sup_{Q_i\in \mcl_i^{\otimes,\mathrm{2G}}(g)} \mu_i(Q_i).
\]

For the reverse inequality, let \(Q_i^\star\) be an optimizer of
\[
\sup_{Q_i\in \mcl_i^{\otimes,\mathrm{2G}}(g)} \mu_i(Q_i)
\]
given by Theorem~\ref{thm:two-group}. Thus, for some
\(k\in\{1,\ldots,n_i-1\}\),
\[
Q_{ij}^\star=\delta_{\mu^+(g)}
\quad \text{for } j\le k,
\qquad
Q_{ij}^\star=(1-g)\delta_1+g\delta_\Gamma
\quad \text{for } j>k.
\]

We first treat the case \(g\in(0,1/2]\). For each \(M\ge 1\), define
\(Q_i(M)=\bigotimes_{j=1}^{n_i}Q_{ij}(M)\) by
\[
Q_{ij}(M)
=
\begin{cases}
\mathrm{Beta}_{[1,\Gamma]}\bigl(M(1-g),Mg\bigr), & j\le k,\\[3pt]
\mathrm{Beta}_{[1,\Gamma]}\bigl(g/M,(1-g)/M\bigr), & j>k.
\end{cases}
\]
The expectations of both distributions lie within the mean band, hence \(Q_i(M)\) belongs to \(\mcl_i^{\otimes,\mathrm{Beta}}(g)\).

Now let \(B_M^+\sim \mathrm{Beta}(M(1-g),Mg)\). Then
\[
\E[B_M^+] = 1-g,
\qquad
\Var(B_M^+) = \frac{g(1-g)}{M+1}.
\]
Therefore \(\Var(B_M^+)\to 0\), so \(B_M^+\to 1-g\) in \(L^2\), hence in distribution. After rescaling,
\[
\mathrm{Beta}_{[1,\Gamma]}\bigl(M(1-g),Mg\bigr)
\Rightarrow
\delta_{\mu^+(g)}.
\]

Next let \(B_M^-\sim \mathrm{Beta}(g/M,(1-g)/M)\). For each integer \(m\ge 1\),
\[
\E\bigl[(B_M^-)^m\bigr]
=
\frac{(g/M)_m}{(1/M)_m},
\]
where \((a)_m=a(a+1)\cdots(a+m-1)\) is the rising factorial. Since
\[
\frac{(g/M)_m}{(1/M)_m}
=
\frac{\frac{g}{M}\left(1+\frac{g}{M}\right)\cdots\left(m-1+\frac{g}{M}\right)}
{\frac{1}{M}\left(1+\frac{1}{M}\right)\cdots\left(m-1+\frac{1}{M}\right)}
\longrightarrow g,
\]
we obtain
\[
\E\bigl[(B_M^-)^m\bigr]\longrightarrow g
\qquad\text{for every } m\ge 1.
\]
Hence, for any polynomial \(p\),
\[
\E[p(B_M^-)]
\longrightarrow
(1-g)p(0)+gp(1).
\]
By density of polynomials in \(C([0,1])\), it follows that
\[
B_M^-
\Rightarrow
(1-g)\delta_0+g\delta_1.
\]
After affine rescaling,
\[
\mathrm{Beta}_{[1,\Gamma]}\bigl(g/M,(1-g)/M\bigr)
\Rightarrow
(1-g)\delta_1+g\delta_\Gamma.
\]

Since the number of coordinates is finite, the product distributions \(Q_i(M)\) converge weakly to \(Q_i^\star\) on \([1,\Gamma]^{n_i}\). Writing
\[
\mu_i(Q_i)=\int \psi_i \, dQ_i,
\]
where \(\psi_i\) is the bounded continuous integrand defining the objective, weak convergence yields
\[
\mu_i(Q_i(M)) \longrightarrow \mu_i(Q_i^\star).
\]
Because each \(Q_i(M)\in \mcl_i^{\otimes,\mathrm{Beta}}(g)\), we conclude that
\[
\sup_{Q_i\in \mcl_i^{\otimes,\mathrm{Beta}}(g)} \mu_i(Q_i)
\;\ge\;
\lim_{M\to\infty}\mu_i(Q_i(M))
=
\mu_i(Q_i^\star)
=
\sup_{Q_i\in \mcl_i^{\otimes,\mathrm{2G}}(g)} \mu_i(Q_i).
\]
Combined with the first inequality, this proves \eqref{eq:sup-beta-equals-sup} for \(g\in(0,1/2]\).

When \(g=0\), Theorem~\ref{thm:two-group} gives the two-group optimizer with marginals \(\delta_\Gamma\) and \(\delta_1\). These are approximated by
\[
\mathrm{Beta}_{[1,\Gamma]}(M,1)\Rightarrow \delta_\Gamma,
\qquad
\mathrm{Beta}_{[1,\Gamma]}(1,M)\Rightarrow \delta_1.
\]
Since the mean-band constraint is vacuous when \(g=0\), the same continuity argument applies and yields \eqref{eq:sup-beta-equals-sup} in this case as well.
\end{proof}

\section{Derivation for Design Sensitivity}\label{app:ds}
This appendix gives the calculations used in Section~\ref{sec:power}. We
restrict to matched pairs and to the paired-difference statistic described in
Section~\ref{subsec:design_setup}. 

Let \(\varrho^{\mathrm{2G}}(g;\Gamma)\) and
\(\varrho^{\mathrm{Bern}}(g;\Gamma)\) denote the least favorable marginal
probability that the higher-score unit in a pair receives treatment under,
respectively, the two-group sensitivity analysis and the Bernoulli sensitivity
analysis.
For the two-group sensitivity analysis, Theorem~\ref{thm:two-group}
implies that
\[
    \varrho^{\mathrm{2G}}(g;\Gamma)
    =
    (1-g)\frac{\mu^+(g)}{\mu^+(g)+1}
    +
    g\frac{\mu^+(g)}{\mu^+(g)+\Gamma}.
\]
For the Bernoulli sensitivity analysis, Theorem~\ref{thm:ber} gives
\[
    \varrho^{\mathrm{Bern}}(g;\Gamma)
    =
    (1-g)^2\frac{\Gamma}{1+\Gamma}
    +
    g(1-g)
    +
    g^2\frac{1}{1+\Gamma}.
\]
When \(g=0\), both stochastic analyses reduce to Rosenbaum's deterministic
sensitivity analysis:
\[
    \varrho^{\mathrm{2G}}(0;\Gamma)
    =
    \varrho^{\mathrm{Bern}}(0;\Gamma)
    =
    \frac{\Gamma}{1+\Gamma}.
\]

Since exactly
one unit is treated in each pair,
\[
    T_i-\bar q_i
    =
    (2Z_{i1}-1)\frac{q_{i1}-q_{i2}}{2},
    \qquad
    |T_i-\bar q_i|
    =
    \frac{q_{i1}-q_{i2}}{2}.
\]
Thus, under a stochastic sensitivity analysis
\(\mathcal A\in\{\mathrm{2G},\mathrm{Bern}\}\), the largest expectation under
the null of \(T_i-\bar q_i\) is
\[
    \bigl(2\varrho^{\mathcal A}(g;\Gamma)-1\bigr)
    |T_i-\bar q_i|.
\]
On the other hand, under the data-generating model in
Section~\ref{subsec:design_setup},
\[
\frac1I\sum_{i=1}^I\left(T_i-\bar q_i\right)
\longrightarrow \E\left[\frac{D_i}{2}\right]
=\frac{\tau}{2},
\qquad
\frac1I\sum_{i=1}^I |T_i-\bar q_i|
\longrightarrow
\E\left[\frac{|D_i|}{2}\right].
\]
Therefore rejection persists asymptotically when
\[
    \frac{\tau}{2}
    >
    \bigl(2\varrho^{\mathcal A}(g;\Gamma)-1\bigr)
    \frac{\E[|D_i|]}{2}.
\]
Equivalently, the asymptotic rejection boundary is
\begin{equation}\label{eq:app-ds-common-boundary}
\varrho^{\mathcal A}(g;\Gamma)
=
\frac12\left(1+\frac{\tau}{\E[|D_i|]}\right).
\end{equation}
For fixed \(g\), the design sensitivity
\(\widetilde\Gamma^{\mathcal A}(\tau;g)\) is the solution of
\eqref{eq:app-ds-common-boundary} in \(\Gamma\). For fixed \(\Gamma\),
\(\widetilde g^{\mathcal A}(\tau;\Gamma)\) is the solution of
\eqref{eq:app-ds-common-boundary} in \(g\).

\subsection{Solving \texorpdfstring{\(\widetilde\Gamma\)}{Gamma} for fixed \texorpdfstring{\(g\)}{g}}

\paragraph{Two-group analysis.}
Substituting \(\varrho^{\mathrm{2G}}(g;\Gamma)\) into
\eqref{eq:app-ds-common-boundary} and clearing denominators gives
\[
    a_2\{\widetilde\Gamma^{\mathrm{2G}}(\tau;g)\}^2
    +
    a_1\widetilde\Gamma^{\mathrm{2G}}(\tau;g)
    +
    a_0
    =0,
\]
where
\[
\begin{aligned}
    a_2
    &=(1-g)\left[
        \frac{2-g}{2}\left(1+\frac{\tau}{\E[|D_i|]}\right)
        -2(1-g)
      \right],\\
    a_1
    &=
      (1+g-g^2)\left(1+\frac{\tau}{\E[|D_i|]}\right)
      -4g(1-g),\\
    a_0
    &=g\left[
        \frac{1+g}{2}\left(1+\frac{\tau}{\E[|D_i|]}\right)
        -2g
      \right].
\end{aligned}
\]
If
\[
\frac12\left(1+\frac{\tau}{\E[|D_i|]}\right)
<
\frac{2(1-g)}{2-g},
\]
then the finite design sensitivity is
\begin{equation}\label{eq:app-ds-2g-gamma}
\widetilde\Gamma^{\mathrm{2G}}(\tau;g)
=
\frac{-a_1-\sqrt{a_1^2-4a_2a_0}}{2a_2}.
\end{equation}
The minus sign before the square root is the relevant branch because, in the
finite case, \(a_2<0\). If the displayed inequality fails, then
\(\varrho^{\mathrm{2G}}(g;\Gamma)\) does not reach the boundary in
\eqref{eq:app-ds-common-boundary} for any finite \(\Gamma\), and we write
\(\widetilde\Gamma^{\mathrm{2G}}(\tau;g)=\infty\).

\paragraph{Bernoulli analysis.}
For the Bernoulli analysis,
\[
\varrho^{\mathrm{Bern}}(g;\Gamma)
=
\frac{\Gamma(1-g)+g}{\Gamma+1}.
\]
Solving \eqref{eq:app-ds-common-boundary} gives
\begin{equation}\label{eq:app-ds-bern-gamma}
\widetilde\Gamma^{\mathrm{Bern}}(\tau;g)
=
\frac{
1+\tau/\E[|D_i|]-2g
}{
1-\tau/\E[|D_i|]-2g
},
\end{equation}
provided
\[
\frac12\left(1+\frac{\tau}{\E[|D_i|]}\right)<1-g.
\]
If this inequality fails, then
\(\widetilde\Gamma^{\mathrm{Bern}}(\tau;g)=\infty\).

When \(g=0\), both \eqref{eq:app-ds-2g-gamma} and
\eqref{eq:app-ds-bern-gamma} reduce to the conventional design sensitivity
\[
    \widetilde\Gamma(\tau;0)
    =
    \frac{1+\tau/\E[|D_i|]}{1-\tau/\E[|D_i|]}
    =
    \frac{\E[|D_i|]+\tau}{\E[|D_i|]-\tau}.
\]

\subsection{Solving \texorpdfstring{\(\widetilde g\)}{g} for fixed \texorpdfstring{\(\Gamma\)}{Gamma}}

\paragraph{Bernoulli analysis.}
For fixed \(\Gamma>1\), solving
\(\varrho^{\mathrm{Bern}}(g;\Gamma)=
\frac12(1+\tau/\E[|D_i|])\) gives
\begin{equation}\label{eq:app-ds-bern-g}
\widetilde g^{\mathrm{Bern}}(\tau;\Gamma)
=
\frac{
\Gamma-(\Gamma+1)\frac12\left(1+\frac{\tau}{\E[|D_i|]}\right)
}{\Gamma-1}
=
\frac{\Gamma-1-(\Gamma+1)\tau/\E[|D_i|]}{2(\Gamma-1)}.
\end{equation}

\paragraph{Two-group analysis.}
For the two-group analysis, substituting \(\Gamma\) into
\(\varrho^{\mathrm{2G}}(g;\Gamma)\) and collecting powers of \(g\) gives
\[
    b_2\{\widetilde g^{\mathrm{2G}}(\tau;\Gamma)\}^2
    +
    b_1\widetilde g^{\mathrm{2G}}(\tau;\Gamma)
    +
    b_0
    =0,
\]
where
\[
\begin{aligned}
b_2
&=
\left\{\frac12\left(1+\frac{\tau}{\E[|D_i|]}\right)-2\right\}
(\Gamma-1)^2,\\
b_1
&=-(\Gamma-1)\left[
(3\Gamma+1)\frac12\left(1+\frac{\tau}{\E[|D_i|]}\right)
-4\Gamma
\right],\\
b_0
&=2\Gamma\left[
(\Gamma+1)\frac12\left(1+\frac{\tau}{\E[|D_i|]}\right)
-\Gamma
\right].
\end{aligned}
\]
The relevant root is the smaller root,
\begin{equation}\label{eq:app-ds-2g-g}
\widetilde g^{\mathrm{2G}}(\tau;\Gamma)
=
\frac{-b_1+\sqrt{b_1^2-4b_2b_0}}{2b_2}.
\end{equation}

The values in \eqref{eq:app-ds-bern-g} and \eqref{eq:app-ds-2g-g} should be
read relative to the admissible range \(0\le g\le 1/2\). If the displayed root
is negative, then the conventional analysis at \(g=0\) already rejects
asymptotically at that value of \(\Gamma\), so the effective threshold is
\(\widetilde g=0\). If the displayed root exceeds \(1/2\), then no admissible
value of \(g\) is large enough to make rejection persist asymptotically at that
value of \(\Gamma\).


\section{Details for the binge drinking study}
\label{app:binge-match}

The binge drinking study in Section~\ref{subsec:binge} used the \texttt{binge} data from the
\texttt{iTOS} package. We restricted the analysis to frequent binge drinkers
and never-bingers, and defined the treatment indicator \(Z=1\) for frequent
binge drinkers and \(Z=0\) for never-bingers.

The match was based on nine pre-treatment baseline covariates: age, sex,
education, body-mass index, waist-to-hip ratio, vigorous activity, current
smoking frequency, smoking-cessation status, and current use of medication for
high blood pressure. Table~\ref{tab:binge-covariates} lists the corresponding
variables in the \texttt{iTOS} data.

\begin{table}[H]
\centering
\begin{tabular}{lp{0.72\textwidth}}
\toprule
Variable & Description \\
\midrule
\texttt{age} & Age in years. \\
\texttt{female} & Indicator for female sex. \\
\texttt{educationf} & Ordered factor for education with five levels:
\(<9\)th grade, 9--11th grade without a high school degree or equivalent,
high school degree or equivalent, some college, and at least a BA degree. \\
\texttt{bmi} & Body-mass index, a measure of obesity. \\
\texttt{waisthip} & Waist-to-hip ratio, a measure of obesity. \\
\texttt{vigor} & Indicator for vigorous activity, either recreational or occupational. \\
\texttt{smokenow} & Integer score for current smoking frequency: Everyday \(<\) SomeDays \(<\) No. \\
\texttt{smokeQuit} & Indicator for having smoked regularly and quit; current smokers and never
smokers both have value 0. \\
\texttt{bpRX} & Indicator for currently taking medication to control high blood pressure. \\
\bottomrule
\end{tabular}
\caption{Baseline covariates used in the NHANES binge-drinking match. Variable
names and definitions follow the \texttt{iTOS} \texttt{binge} data
documentation.}
\label{tab:binge-covariates}
\end{table}
We constructed a rank-based Mahalanobis distance on the expanded covariate
matrix, with education treated as a factor. We also imposed a propensity-score
caliper of \(0.08\) standard deviations on the logit propensity score. The
propensity score was estimated by logistic regression of the treatment
indicator on the same baseline covariates. 
Because the number of never-bingers
was much larger than the number of frequent binge drinkers, restricted full
matching was implemented using \texttt{optmatch::fullmatch} with
\texttt{omit.fraction = 0.7}, requiring at least one control per treated unit
and allowing at most ten controls per treated unit.

The final match contained 206 matched sets, each with one treated unit. The
numbers of controls per matched set ranged from 1 to 10, with counts
\[
45,\ 22,\ 14,\ 13,\ 13,\ 7,\ 7,\ 4,\ 5,\ 76
\]
for 1 through 10 controls, respectively. Thus the matched sample contained
1,382 individuals: 206 treated units and 1,176 controls.

Because \texttt{educationf} is an ordered factor with five levels, it is
represented in the matching distance by four numerical contrast variables. In
Figure~\ref{fig:balance}, the imbalance for education is reported as the
largest absolute standardized difference among these four variables.


\end{document}